\documentclass[runningheads,11pt]{llncs}
\usepackage[T1]{fontenc}
\usepackage[utf8]{inputenc}
\usepackage{graphicx}
\usepackage{amssymb}
\usepackage{comment}
\usepackage{mathtools}
\usepackage{caption}
\usepackage{float}
\usepackage{tikz}
\usetikzlibrary{shapes.geometric}
\usepackage[ruled,linesnumbered]{algorithm2e}

\usepackage[papersize={8.5in,11in},margin=2.5cm]{geometry}

\usepackage{proof-at-the-end}
\pgfkeys{/prAtEnd/disc/.style={%
    proof at the end, end, restate, no link to proof }%
}

\usepackage[cmtip,all]{xy}

\newcommand{\longsquiggly}{\xymatrix{{}\ar@{~>}[r]&{}}}

\newcommand{\oneRound}
{\xymatrix@C=1.5em{{}\ar@{~>}[r]&{}}}

\newcommand{\oneRoundNoCrash}
{\overset{\text{no crash}}
{\xymatrix@C=3.5em{{}\ar@{~>}[r]&{}}}}

\newcommand{\oneRoundSecCrash}
{\overset{\text{sec. crash}}
{\xymatrix@C=4.5em{{}\ar@{~>}[r]&{}}}}

\newcommand{\oneRoundMainCrash}
{\overset{\text{main crash}}
{\xymatrix@C=4.5em{{}\ar@{~>}[r]&{}}}}

\newcommand{\manyRounds}{\overset{\text{*}}{\xymatrix@C=1.5em{{}\ar@{~>}[r]&{}}}}

\usepackage{xspace}
\newcommand{\ie}{\emph{i.e., \xspace}}

\newcommand{\rot}{\texttt{rot}}

\newcommand{\vtarget}{\ensuremath{v_{\mathit{target}}}\xspace}

\newcommand{\confAxis}[4][]{
\draw[#1] ({360/\n * #2 - 90}:\radius+#3) --  ({360/\n * (\n/2 + #2)-90}:\radius+#4);
}

\newcommand{\confAngle}[1]{
360/\n * #1  -90
}

\newcommand{\confRing}{
\draw (\confAngle{0}:\radius) arc (\confAngle{0}:\confAngle{\n}:\radius);
\foreach \s in {0,...,\n}
{
  \node[draw, fill=white, circle,inner sep=0.07cm] (r-\s) at (\confAngle{\s}:\radius) {};
}
}

\newcommand{\confRobot}[2][]{
\node[robot,#1]  at (\confAngle{#2}:\radius) {};
}

\newcommand{\minipageFigure}[6]{
\begin{minipage}{#1\textwidth}\centering
\begin{tikzpicture}
\tikzstyle{robot}=[circle,fill=black!100,inner sep=0.05cm]
  
\def\n{#2}
\def\radius{#3}

\confRing

#6

\end{tikzpicture}
\captionof{figure}{#4} #5
\end{minipage}
}

\newcommand{\NE}{\ensuremath{\mathit{NE}}\xspace}
\newcommand{\LTwo}{\ensuremath{\mathit{L2}}\xspace}
\newcommand{\LThree}{\ensuremath{\mathit{L3}}\xspace}
\newcommand{\LFour}{\ensuremath{\mathit{L4}}\xspace}
\newcommand{\LFive}{\ensuremath{\mathit{L5}}\xspace}
\newcommand{\QNE}{\ensuremath{\mathit{QNE}}\xspace}
\newcommand{\Periodic}{\ensuremath{\mathit{P}}\xspace}

\newcommand\MoveMain{\texttt{M$_{\texttt{main}}$}\xspace}
\newcommand\MoveSecondary{\texttt{M$_{\texttt{secondary}}$}\xspace}

\newcommand\MoveLTwo{\texttt{M$_{\LTwo}$}\xspace}
\newcommand\MoveOpposite{\texttt{M$_{\texttt{opp}}$}\xspace}
\newcommand\MoveSame{\texttt{M$_{\texttt{same}}$}\xspace}

\newcommand{\LThreeFigure}{
\begin{tikzpicture}
    \draw (-120:2) arc (-120:-60:2);
    \node[draw, fill=white, circle,inner sep=0.07cm] (r1) at (-110:2) {};
    \node[draw, fill=black, circle,inner sep=0.05cm] (r1) at (-110:2) {};
    \node[draw, fill=white, circle,inner sep=0.07cm] (r1) at (-100:2) {};
    \node[draw, fill=white, circle,inner sep=0.07cm] (r1) at (-90:2) {};
    \node[draw, fill=black, circle,inner sep=0.05cm] (r1) at (-90:2) {};
    \node[draw, fill=white, circle,inner sep=0.07cm] (r1) at (-80:2) {};
    \node[draw, fill=white, circle,inner sep=0.07cm] (r1) at (-70:2) {};
    \node[draw, fill=black, circle,inner sep=0.05cm] (r1) at (-70:2) {};
\end{tikzpicture}
}

\newcommand{\LThreeFigureP}{
\begin{tikzpicture}
    \draw (-120:2) arc (-120:-60:2);
    \node[draw, fill=white, circle,inner sep=0.07cm] (r1) at (-110:2) {};
    \node[draw, fill=black, circle,inner sep=0.05cm] (r1) at (-110:2) {};
    \node[draw, fill=white, circle,inner sep=0.07cm] (r1) at (-100:2) {};
    \node[draw, fill=white, circle,inner sep=0.07cm] (r1) at (-90:2) {};
    \node[draw, fill=black, circle,inner sep=0.05cm] (r1) at (-90:2) {};
    \node[draw, fill=white, circle,inner sep=0.07cm] (r1) at (-80:2) {};
    \node[draw, fill=black, circle,inner sep=0.05cm] (r1) at (-80:2) {};
    \node[draw, fill=white, circle,inner sep=0.07cm] (r1) at (-70:2) {};
\end{tikzpicture}
}

\newcommand{\LThreeFigurePS}{
\begin{tikzpicture}
    \draw (-120:2) arc (-120:-60:2);
    \node[draw, fill=white, circle,inner sep=0.07cm] (r1) at (-110:2) {};
    \node[draw, fill=white, circle,inner sep=0.07cm] (r1) at (-100:2) {};
    \node[draw, fill=black, circle,inner sep=0.05cm] (r1) at (-100:2) {};
    \node[draw, fill=white, circle,inner sep=0.07cm] (r1) at (-90:2) {};
    \node[draw, fill=black, circle,inner sep=0.05cm] (r1) at (-90:2) {};
    \node[draw, fill=white, circle,inner sep=0.07cm] (r1) at (-80:2) {};
    \node[draw, fill=black, circle,inner sep=0.05cm] (r1) at (-80:2) {};
    \node[draw, fill=white, circle,inner sep=0.07cm] (r1) at (-70:2) {};
\end{tikzpicture}
}

\newcommand{\LFourFigure}{
\begin{tikzpicture}
    \draw (-120:2) arc (-120:-60:2);
    \node[draw, fill=white, circle,inner sep=0.07cm] (r1) at (-110:2) {};
    \node[draw, fill=black, circle,inner sep=0.05cm] (r1) at (-110:2) {};
    \node[draw, fill=white, circle,inner sep=0.07cm] (r1) at (-100:2) {};
    \node[draw, fill=black, circle,inner sep=0.05cm] (r1) at (-100:2) {};
    \node[draw, fill=white, circle,inner sep=0.07cm] (r1) at (-90:2) {};
    \node[draw, fill=white, circle,inner sep=0.07cm] (r1) at (-80:2) {};
    \node[draw, fill=black, circle,inner sep=0.05cm] (r1) at (-80:2) {};
    \node[draw, fill=white, circle,inner sep=0.07cm] (r1) at (-70:2) {};
    \node[draw, fill=black, circle,inner sep=0.05cm] (r1) at (-70:2) {};
\end{tikzpicture}
}

\newcommand{\LFourPrimeFigure}{
\begin{tikzpicture}
    \draw (-120:2) arc (-120:-60:2);
    \node[draw, fill=white, circle,inner sep=0.07cm] (r1) at (-110:2) {};
    \node[draw, fill=black, circle,inner sep=0.05cm] (r1) at (-110:2) {};
    \node[draw, fill=white, circle,inner sep=0.07cm] (r1) at (-100:2) {};
    \node[draw, fill=black, circle,inner sep=0.05cm] (r1) at (-100:2) {};
    \node[draw, fill=white, circle,inner sep=0.07cm] (r1) at (-90:2) {};
    \node[draw, fill=black, circle,inner sep=0.05cm] (r1) at (-90:2) {};
    \node[draw, fill=white, circle,inner sep=0.07cm] (r1) at (-80:2) {};
    \node[draw, fill=white, circle,inner sep=0.07cm] (r1) at (-70:2) {};
    \node[draw, fill=black, circle,inner sep=0.05cm] (r1) at (-70:2) {};
\end{tikzpicture}
}

\newcommand{\LFiveFigure}{
\begin{tikzpicture}
    \draw (-120:2) arc (-120:-60:2);
    \node[draw, fill=white, circle,inner sep=0.07cm] (r1) at (-110:2) {};
    \node[draw, fill=black, circle,inner sep=0.05cm] (r1) at (-110:2) {};
    \node[draw, fill=white, circle,inner sep=0.07cm] (r1) at (-100:2) {};
    \node[draw, fill=black, circle,inner sep=0.05cm] (r1) at (-100:2) {};
    \node[draw, fill=white, circle,inner sep=0.07cm] (r1) at (-90:2) {};
    \node[draw, fill=black, circle,inner sep=0.05cm] (r1) at (-90:2) {};
    \node[draw, fill=white, circle,inner sep=0.07cm] (r1) at (-80:2) {};
    \node[draw, fill=black, circle,inner sep=0.05cm] (r1) at (-80:2) {};
    \node[draw, fill=white, circle,inner sep=0.07cm] (r1) at (-70:2) {};
    \node[draw, fill=black, circle,inner sep=0.05cm] (r1) at (-70:2) {};
\end{tikzpicture}
}

\newcommand{\LTwoFigure}{
\begin{tikzpicture}
    \draw (-120:2) arc (-120:-60:2);
    \node[draw, fill=white, circle,inner sep=0.07cm] (r1) at (-110:2) {};
    \node[draw, fill=white, circle,inner sep=0.07cm] (r1) at (-100:2) {};
    \node[draw, fill=black, circle,inner sep=0.05cm] (r1) at (-100:2) {};
    \node[draw, fill=white, circle,inner sep=0.07cm] (r1) at (-90:2) {};
    \node[draw, fill=black, circle,inner sep=0.05cm] (r1) at (-90:2) {};
    \node[draw, fill=white, circle,inner sep=0.07cm] (r1) at (-80:2) {};
    \node[draw, fill=white, circle,inner sep=0.07cm] (r1) at (-70:2) {};
\end{tikzpicture}
}

\newcommand{\SK}[1]{\noindent\textcolor{pink}{{\fontfamily{phv}\selectfont SK-NOTE: #1}}}

\def\qed{\hfill$\Box$\par\vskip1em}

\begin{document}
\title{Stand-Up Indulgent Gathering on Rings\thanks{This paper was supported by JSPS KAKENHI No. 23H03347, 23K11059, and ANR project SAPPORO (Ref. 2019-CE25-0005-1).}}
\titlerunning{Stand-Up Indulgent Gathering on Rings}

\author{Quentin Bramas\inst{1}%\orcidID{0000-0003-0612-5616} 
\and Sayaka Kamei\inst{2}%\orcidID{0000-0003-1716-3028}
\and Anissa Lamani\inst{1}%\orcidID{0000-0001-7774-8402} 
\and\\ 
S\'ebastien Tixeuil \inst{3}%\orcidID{0000-0002-0948-7172}
}

\institute{
University of Strasbourg, ICube, CNRS, France. 
\and
Hiroshima University, Japan.
\and
Sorbonne University, CNRS, LIP6, IUF, France. }

\maketitle

\begin{abstract}
We consider a collection of $k \geq 2$ robots that evolve in a ring-shaped network without common orientation, and address a variant of the crash-tolerant gathering problem called the \emph{Stand-Up Indulgent Gathering} (SUIG): given a collection of robots, if no robot crashes, robots have to meet at the same arbitrary location, not known beforehand, in finite time; if one robot or more robots crash on the same location, the remaining correct robots gather at the location of the crashed robots. We aim at characterizing the solvability of the SUIG problem without multiplicity detection capability. 
\end{abstract}

%\newpage
\section{Introduction}

Mobile robotic swarms have recently gained significant interest within the Distributed Computing scientific community. For over two decades, researchers have aimed to characterize the precise assumptions that enable basic problems to be solved for robots represented as disoriented (each robot has its own coordinate system), oblivious (robots cannot remember past actions), and dimensionless points evolving in Euclidean space. One key assumption is the scheduling assumption~\cite{PGN2019}, where robots can execute their protocol fully synchronized (FSYNC), completely asynchronously (ASYNC), or with a fairly chosen subset of robots scheduled for synchronous execution (SSYNC).

The \emph{gathering} problem~\cite{SY1999} serves as a benchmark due to its simplicity of expression (robots must gather at an unknown location in finite time) and its computational tractability (two SSYNC-scheduled robots cannot gather without additional assumptions).

As the number of robots increases, so does the probability of at least one failure. However, relatively few works consider the possibility of robot failures. One simple failure is the \emph{crash} fault, where a robot unpredictably stops executing its protocol. In gathering, one must prescribe the expected behavior in the presence of crash failures. Two variants have been studied: \emph{weak} gathering expects all correct (non-crashed) robots to gather regardless of crashed robot positions, while \emph{strong} gathering (also known as stand-up indulgent gathering -- SUIG) expects correct robots to gather at the unique crash location.
In continuous Euclidean space, weak gathering is solvable in the SSYNC model~\cite{ND2006,ZSS2013,QS2015,XMP2020}, while SUIG (and its two-robot variant, stand-up indulgent rendezvous -- SUIR) is only solvable in the FSYNC model~\cite{QAS2020,QAS2021,BLT23j}.

A recent trend~\cite{PGN2019} has been to move from a continuous environment to a discrete one. In the discrete setting, robots can occupy a finite number of locations and move between neighboring locations. This neighborhood relation is conveniently represented by a graph whose nodes are locations, leading to the ``robots on graphs'' denomination. The discrete setting is better suited for describing constrained physical environments or environments where robot positioning is only available from discrete sensors~\cite{TPRLSX19}. From a computational perspective, the continuous and discrete settings are unrelated: on one hand, the number of possible configurations (i.e., robot positions) is much more constrained in the discrete setting than in the continuous setting (only a finite number of configurations exist in the discrete setting); on the other hand, the continuous setting offers algorithm designers more flexibility to solve problematic configurations (e.g., using arbitrarily small movements to break symmetry).
To our knowledge, the only previous work considering SUIR and SUIG in a discrete setting is due to Bramas et al.~\cite{techreport_line}. They provide a complete characterization of the problems solvability for line-shaped networks. More specifically, they show that the problem is unsolvable in SSYNC, and provide a solution in FSYNC.

In this paper, we consider the discrete setting and aim to further characterize the solvability of the SUIR and SUIG problems when the visibility radius of robots is unlimited. In a set of locations whose neighborhood relation is represented by a ring-shaped graph, robots must gather at a single unknown location. Furthermore, if one or more robots crash at the same location, all robots must gather at that location. With respect to previous work~\cite{techreport_line}, a ring topology induces a lot more possibility for symmetric situations, and a robot may no longer use the knowledge of being an edge robot (that is, being located at an extremal position with respect to all other robots on the line).
We present the following results. With two robots, the SUIR problem is unsolvable in SSYNC, and solvable in FSYNC only when the initial configuration is not periodic and node-node symmetric. 
With three robots and more, impossibility results for SSYNC SUIR and FSYNC gathering (from a periodic or edge-edge symmetric configuration) extend to SUIG. Also, we prove that three robots cannot solve SUIG in FSYNC when the size of the ring is multiple of $3$. 
On the positive side, we present a SUIG algorithm for more than three robots and an odd $n$ in FSYNC, assuming an initial node-edge symmetric configuration that is not periodic. Among other difficulties, the key technique enabling the solution is to have at least two robots on distinct locations move in any situation (to withstand a crash), that can be of independent interest for further developing stand-up indulgent solutions.

The sequel of the paper is organized as follows: Section~\ref{sec:model} formalizes our execution model, Section~\ref{sec:suir} presents results for the case of two robots, while Section~\ref{sec:suig} concentrates on the main case. Finally, concluding remarks are provided in Section~\ref{sec:conclusion}.

\section{Model}
\label{sec:model}
The rings we consider consist of $n$ unlabeled nodes $u_0, u_1, u_2, \dots, u_{n-1}$ such that $u_i$ is connected to both $u_{((i-1)\mod n)}$ and $u_{((i+1)\mod n)}$. 

Let $R=\{r_1, r_2, \dots, r_k\}$ denote the set of $k\geq 2$ autonomous robots.
Robots are assumed to be anonymous (i.e., they are indistinguishable), uniform (i.e., they all execute the same program, and do not make use of localized parameters), oblivious (i.e., they cannot remember their past actions), and disoriented (i.e., they have no common sense of direction).
We assume that robots do not know $k$ (the number of robots).
In addition, they are unable to communicate directly, however, they have the ability to sense their environment, including the position of the other robots. Based on the configuration resulting of the sensing,
they decide whether to move or to stay idle. Each robot $r$ executes cycles infinitely many times: \emph{(i)} first, $r$ takes a snapshot of its environment to see the position of the other robots (LOOK phase), \emph{(ii)} according to the snapshot, $r$ decides to move or not to move (COMPUTE phase), and \emph{(iii)} if $r$ decides to move, it moves to one of its neighbor nodes depending on the choice made during the COMPUTE phase (MOVE phase). We call such cycles \emph{LCM cycles}.
We consider the \emph{FSYNC} model in which at each time instant $t$, called \emph{round}, each robot executes a LCM cycle synchronously with all other robots, and the \emph{SSYNC} model, where a non-empty subset of robots chosen by an adversarial scheduler executes a LCM cycle synchronously, at each time instant $t$.

A node $u$ is \emph{occupied} if it hosts at least one robot, otherwise, $u$ is \emph{empty}. We say that there is a \textit{tower} (or a \textit{multiplicity}) on a node $u$ if $u$ hosts (strictly) more than one robot. The ability for the robots to identify nodes that host a tower is called \textit{multiplicity detection}. 
When robots can detect a tower, we say that they have a \textit{strong} multiplicity detection if they know the exact number of robots participating to the tower. Otherwise, they have a \textit{weak} multiplicity detection capability. 
In this work, we assume that the robots have no multiplicity detection, \emph{i.e.}, they cannot decide whether an occupied node is a tower.

During the execution, robots move and occupy some nodes, their positions form the \emph{configuration} $C_t=(d(u_0), d(u_1), \dots)$ of the system at time $t$, where $d(u_i) = 0$ if the node $u_i$ is empty and $d(u_i) = 1$ if $u_i$ is occupied.
Initially, we assume that the configuration is distinct, \emph{i.e.}, each node is occupied by at most one robot. 
For the analysis, a robot $r$ may also denotes the node occupied by $r$.

A sequence of consecutive occupied nodes is a \emph{block}.
Similarly, a sequence of consecutive empty nodes is a \emph{hole}.

The \emph{distance} between two nodes $u_i$ and $u_j$ is the number of edges in a shortest path connecting them. The distance between two robots $r_i$ and $r_j$ is the distance between the two nodes occupied by $r_i$ and $r_j$, respectively.
We denote the distance between $u_i$ and $u_j$ (resp. $r_i$ and $r_j$) $dist(u_i,u_j)$ (resp. $dist(r_i,r_j)$).
Two robots or two nodes are \emph{neighbors} if the distance between them is one.

An \emph{algorithm} $A$ is a function mapping the snapshot (obtained during the LOOK phase) to a neighbor node (or the same node should the robot decide not to move) destination to move to (during the MOVE phase).  
An \emph{execution} ${\cal E}=(C_0, C_1,\dots)$ of $A$ is a sequence of configurations, where $C_0$ is an initial configuration, and every configuration $C_{t+1}$ is obtained from $C_{t}$ by applying $A$. 

The \emph{angle} between two nodes $u_i$ and $u_j$ is determined by the number of nodes in the path that connects them in either a clockwise or counterclockwise direction. This angle can be expressed as either the distance or its $n$-complement. For our analysis, we assume a fixed direction that is unknown to the robots and can be arbitrary unless otherwise specified. The rotation of an angle $x$ is denoted by $\rot(x)$, which maps each node to the $x$-th node in a specific direction. It's important to note that $\rot(x)$ is equivalent to a rotation of angle $2x\pi/n$ radians in the same direction when mapping the $n$-node ring to a circle in Euclidean space with evenly spaced nodes. The image of node $u_i$ after rotation $\rot(x)$ is represented by $\rot(x)(u_i)$.

An automorphism is a bijection $f$ of nodes such that $f(u_i)$ and $f(u_j)$ are neighbors if and only if $u_i$ and $u_j$ are neighbors. In an $n$-ring, there are two types of automorphisms: rotational symmetries and reflections. These depend on whether the angles between nodes are preserved or not. The group of symmetries is the dihedral group $D_{n}$, which contains $n$ rotational symmetries and $n$ reflections. A reflection can be classified as an \emph{edge-edge symmetry}, \emph{node-edge symmetry}, or \emph{node-node symmetry} depending on whether it has no fixed points, exactly one fixed point, or exactly two fixed points, respectively. The axis of a reflection is the imaginary line that passes through the fixed points and/or fixed edges. This axis corresponds to the axis of symmetry when the ring is mapped to the Euclidean plane. In the remainder of this text, a symmetry always refers to a reflection.
 
We say that a set $S$ of nodes is \emph{symmetric}, or contains an axis of symmetry, if there is a reflection $f$ such that $f(S) = S$. We also say that a set $S$ of nodes is \emph{rotational symmetric} or \emph{periodic} if there is a rotation $f$ such that $f(S) = S$.

Configurations that have no tower are classified into three classes~\cite{REA2008}. Let $C$ be a configuration, and let $S$ be the set of occupied nodes in $C$. 
Then, $C$ is \emph{symmetric} if $S$ is symmetric. 
Also, $C$ is \emph{periodic} if there exists a non-trivial rotational symmetry $f$ such that $f(S)= S$ (it is represented by a configuration of at least two copies of a sub-sequence). Otherwise, $C$ is \emph{rigid}.

We define the view of a robot $r$ located on node $u_i$ as the pair $\{S^+,S^-\}$ ordered in the lexicographic order where $S^+(t) = d(u _i), d(u_{i+1}), \dots, d(u_{i+n-1})$ and $S^-(t) = d(u_i), d(u_{i-1}), \dots, d(u_{i-(n-1)}).$

\subsection*{Problem definition}
A robot is said to be \emph{crashed} at time $t$ if it is not activated at any time $t^\prime \geq t$.
That is, a crashed robot stops execution and remains at the same position indefinitely.
We call the node where robots crash \emph{crashed node}.
We assume that robots cannot distinguish the crashed node in their snapshots (i.e., they are able to see the crashed robots, but they remain unaware of their crashed status).
%Since we consider oblivious robots, we assume that the crashes, if any, occur at the start of the execution.
A crash, if any, can occur at any round of the execution. Furthermore, if more than one crash occur, all crashes occur at the same location. In our model, since robots do not have multiplicity detection capability, a location with a single crashed robot and with multiple crashed robots are indistinguishable, and are thus equivalent. In the sequel, for simplicity, we consider at most one crashed robot.

We consider the \emph{Stand Up Indulgent Gathering} (SUIG) problem.
An algorithm solves the SUIG problem if, for any initial configuration $C_0$ with up to one crashed node, and for any execution ${\cal E}=(C_0, C_1,\dots)$, there exists a round $t$ such that all robots (including the crashed one) gather at a single node, not known beforehand, for all $t^\prime\geq t$.
We call the special case with $k=2$, the \emph{Stand Up Indulgent Rendezvous} (SUIR) problem.
Observe that in the case of multiple crashed nodes, the SUIG problem becomes unsolvable. Hence, in the sequel, we assume that if several robots crash, they crash on the same node. So, overall, there is at most one crashed node.

%%%%%%%%%%%%%%%%%%%%%%
%
%           RING
%
%%%%%%%%%%%%%%%%%%%%%%%

%\section{Ring-shaped networks}
%\label{sec:ring}

%We consider in this section robots evolving in a ring-shaped network. As for Section \ref{sec:line}, we address both the rendezvous and the gathering problem. 

\section{Stand Up Indulgent Rendezvous}
\label{sec:suir}

In this section, we consider the case where two robots occupy initially distinct nodes. 
We first recall an impossibility result for gathering on rings as Theorem~\ref{thm:FSYNC-edge-ring-impossible}.
\begin{theorem}[\cite{REA2008}]\label{thm:FSYNC-edge-ring-impossible}
The gathering problem is unsolvable in FSYNC on ring networks starting from a periodic or an edge-edge symmetric configuration, even with strong multiplicity detection, even with two robots.
\end{theorem}
A direct consequence of this theorem is that SUIR is also impossible in the same setting (indeed, when there are no crashes, SUIR must guarantee rendezvous).  

Next, we consider initial configurations that are node-edge symmetric (with $n$ being odd). 
%When no fault occurs, the two robots have to move towards the hole with an odd size to achieve rendezvous. However, if one of them crashes, whenever the other robot moves, it must apply the same behavior (move toward the hole with an odd size), and hence move back to its previous position. Using this argument, we can show that:}
\begin{lemma}
%\begin{theoremEnd}[disc]{lemma}
\label{lem:FSYNC-edge-node-ring-impossible}
On rings, starting from a node-edge symmetric configuration, there exists no deterministic algorithm that solves the SUIR problem by two oblivious robots without additional hypothesis, even in FSYNC.
\end{lemma}
%\end{theoremEnd}
%\begin{proofEnd}
\begin{proof}
Assume for the purpose of contradiction that such an algorithm exists. 
As robots are oblivious and have no common sense of direction, they can either move towards each other or move away from each other.

Observe first that an even sized ring cannot be node-edge symmetric. So, we consider the case when the size of ring $n$ is odd (then, every configuration is node-edge symmetric). 
When no fault occurs, the robots have to move towards 
the hole with an odd size to achieve rendezvous.  
However, if one of them crashes, whenever the other robot moves, it must apply the same behavior (move toward the hole with an odd size), and hence move back to its previous position. That is, the non-crashed robot moves back and forth between two nodes forever. As a result, the two robots never meet, a contradiction. \qed
%\end{proofEnd}
\end{proof}

The remaining case to address is when the configuration is not periodic but node-node symmetric. Observe that such a symmetry only happens in even-sized rings. 
In this case, there is a simple strategy in FSYNC, which consists in making the robots move towards each other taking the shortest path (as the configuration is not periodic, there always exists a shortest path and a longest path between the two robots). 
So, if none of the two robots crashes, then the robots eventually meet on a single node (that is on the initial axis of symmetry). 
Otherwise, as a single robot moves, the smallest distance between the two robots becomes odd. However, the correct robot can deduce that the other has crashed, and hence continues to move towards the other robot through the shortest path. Hence, the distance between the two robots decreases at each round, and the two robots are eventually located on two adjacent nodes. Finally, to achieve the rendezvous, the non-crashed robot simply moves to the adjacent occupied node. 

In the SSYNC model, SUIR remains impossible, even starting from a node-node symmetric configuration that is not periodic. %\textcolor{red}{
Indeed, as both robots need to move at each time instant, if we assume a node-node symmetric configuration where a single robot $r$ is activated, whatever the direction chosen by $r$, a node-edge symmetric configuration is reached. By Theorem \ref{thm:FSYNC-edge-ring-impossible} and Lemma \ref{lem:FSYNC-edge-node-ring-impossible}, we can deduce: %}

\begin{lemma}
%\begin{theoremEnd}[disc]{lemma}
\label{lem:FSYNC-node-node-ring-impossible}
On rings, starting from a node-node symmetric distinct configuration that is not periodic, there exists no SUIR algorithm in SSYNC, without additional hypotheses.
%\end{theoremEnd}
\end{lemma}
\begin{proof}
%\begin{proofEnd}
Suppose for the sake of contradiction that there exists such an algorithm $A$.
%\textcolor{red}{
If the configuration is \emph{only} node-node symmetric and not periodic, the distance between the two robots in both direction is always even, but they are different.%} 
%\textcolor{red}{AL: if the two robots are not on the nodes that are on the axis of symmetry, the distance between them in each direction is always even. The configuration is not periodic if these distances are different? } . Indeed, if both distances (on each direction) are odd, then the configuration is periodic. By contrast, if the distance in only one direction is odd, then the configuration is node-edge symmetric, hence it is not node-node symmetric. 
Since $A$ solves SUIR problem, robots cannot remain idle. Consider a schedule where from this configuration a single robot is activated, and thus moves in either direction. We reach a configuration in which the distance between the two robots in each direction is odd. 
Hence, the configuration is either periodic or edge-edge symmetric. By Theorem~\ref{thm:FSYNC-edge-ring-impossible}, algorithm $A$ cannot solve SUIR from this configuration, a contradiction.\qed
%\end{proofEnd}
\end{proof}

Overall, SUIR is feasible in FSYNC in even-sized rings, starting from a non-periodic, node-node symmetric configuration, and impossible in all other settings.

\section{Stand-Up Indulgent Gathering}
\label{sec:suig}

In this section, we first extend existing impossibility results for SSYNC gathering and SSYNC SUIR to show that SUIG is not solvable in rings in SSYNC without additional hypothesis, which justifies our later assumption to only consider FSYNC. 

%\ST{The case of node-node symmetric and rigid are missing. This section lacks a main theorem stating the algorithm works with some conditions in FSYNC.}
%QB: what if we start without symmetry?

%\subsection{More Impossibilities}

%\textcolor{red}{
In SSYNC, we can show that there always exists an execution where a configuration with only two occupied nodes is reached. As robots have no multiplicity detection, from such a configuration, it is known that the gathering is impossible. That is:%}

%\begin{theoremEnd}[disc]{lemma}
\begin{lemma}
\label{lem:no-SSYNC-ring}
In rings, the SUIG problem is unsolvable in SSYNC without additional hypotheses.
%\end{theoremEnd}
\end{lemma}
\begin{proof}
%\begin{proofEnd}
From Theorem~\ref{thm:FSYNC-edge-ring-impossible}, Lemmas~\ref{lem:FSYNC-edge-node-ring-impossible} and \ref{lem:FSYNC-node-node-ring-impossible}, we know that the SUIR is not solvable in SSYNC. As we assume that robots do not have multiplicity detection, in the execution of every algorithm that solves the SUIG problem, no configuration shall contain exactly two occupied nodes. Indeed, the SSYNC adversary may activate all robots on the same node synchronously, making them behave as a single robot, leading to the previous impossibility case with two robots. 
Now, assume for the purpose of contradiction that there exists an algorithm $A$ that solves the SUIG problem, and let $C$ be the last configuration in the execution of $A$ where the number of occupied nodes is strictly larger than 2 (i.e., the configuration reached just before gathering is achieved). As robots can only move by one edge at each round, the number of occupied nodes in $C$ is 3. As we consider the SSYNC model, the adversary can activate only the robots located on a single node. That is, a configuration $C'$ in which the number of occupied nodes is equal to 2 is reached. A contradiction. \qed
%\end{proofEnd}
\end{proof}

%\textcolor{red}{When $k=3$ and the configuration is symmetric, there is one robot that is located on the axis of symmetry. By assuming that this robot is the robot that crashes, the only way to achieve the gathering is to make the other robots move towards the crashed robot. When $n\mod 3 = 0$, there exists a non empty set of configurations from which a periodic configuration is reached. Hence:  }
\begin{lemma}
%\begin{theoremEnd}[disc]{lemma}
There exists no algorithm in FSYNC that solves the SUIG problem by $k=3$ oblivious robots when $n\mod 3 = 0$.
%\end{theoremEnd}
\end{lemma}
\begin{proof}
%\begin{proofEnd}
Assume by contradiction that the Lemma does not hold and let consider the configuration shown in Figure \ref{fig:impossible-r3}. As the robot that is on the axis of symmetry could be the faulty one, at least one of the two other robots should also be allowed to move. Since the configuration is symmetric, we cannot distinguish between the two other robots. That is both of them are allowed to move. If they move in the opposite side of the robot on the axis, eventually they become neighbors, and therefrom they can only exchange their positions. That is, in order to perform the gathering, the algorithm needs to have a rule making the two robots move towards the robot on the axis of symmetry. By moving a periodic configuration is reached. 
As robots cannot distinguish both of their sides of the ring, %. The scheduler is one that chooses the destination to take. It can hence, 
the two robots move back to their respective previous position. That is, the robots move back and forth forever. A contradiction. \qed
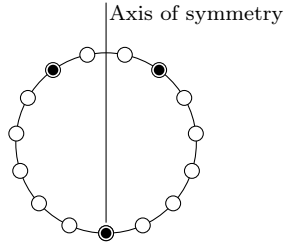
\begin{figure}[hb]
\begin{center}
\begin{tikzpicture}
\tikzstyle{robot}=[circle,fill=black!100,inner sep=0.05cm]
  
\def\n{15}
\def\radius{1.2cm}

\confRing

\confAxis{0}{-4}{19}

\confRobot{-6}
\confRobot{6}
\confRobot{0}

\node[] at (1.2,1.7) {\scriptsize Axis of symmetry};

\end{tikzpicture}
\caption{An instance of an initial configuration with $k=3$ when $n\mod 3 = 0$}\label{fig:impossible-r3}
        \end{center}
\end{figure}
\end{proof}
%\end{proofEnd}

As the gathering is known to be impossible from a periodic or an edge-edge symmetric configuration~\cite{REA2008}, 
%\subsection{An algorithm for initial node-edge symmetry} %\\ \textcolor{red}{($k>3$ or $k>2 \land n\mod 3\neq 0$)}} AL: moved it to the first paragraph
we propose an algorithm which solves the SUIG problem on ring-shaped networks of size $n$ by $k$ robots such that $k>3$ or ($k>2 \land n\mod 3\neq 0$) assuming node-edge symmetric initial configurations. In the sequel, we use $r$ to represent both a robot $r$ and its occupied node.

\begin{figure}[b]
 \begin{minipage}{0.48\textwidth}
\begin{center}
\begin{tikzpicture}
\tikzstyle{robot}=[circle,fill=black!100,inner sep=0.05cm]
  
\def\n{21}
\def\radius{1.2cm}

\confRing

\confAxis{0}{-4}{20}

\confRobot{-10}
\confRobot{10}
\confRobot[blue]{-9}
\confRobot[blue]{9}
\confRobot{-6}
\confRobot{6}
\confRobot[green]{-2}
\confRobot[green]{2}

\node[] (targettxt) at (\confAngle{-2.5}:2.2) {\vtarget};

\node[] at (1.3,1.8) {\scriptsize Axis of symmetry};

\node[] (mr) at (\confAngle{3.5}:2.7) {\scriptsize main robots};
\node[] (sr) at (\confAngle{7}:2.8) {\scriptsize secondary robots};

\draw[-latex] (targettxt) to [bend right=15] (r-0) ;
\draw[-latex] (mr) to [bend left=10] (r-2) ;
\draw[-latex] (sr) to [bend right=10] (r-9) ;

\end{tikzpicture}
\captionof{figure}{An instance of a node-edge symmetric configuration}\label{fig:example-symmetric}
        \end{center}
\end{minipage}
~~~~\begin{minipage}{0.48\textwidth}
\begin{center}
\begin{tikzpicture}
\tikzstyle{robot}=[circle,fill=black!100,inner sep=0.05cm]
  
\def\n{21}
\def\radius{1.2cm}

\confRing

\confAxis{0}{-4}{20}

\confRobot{-10}
\confRobot{10}
\confRobot{-9}
\confRobot{9}
\confRobot{-6}
\confRobot{6}
\confRobot[]{-2}
\confRobot[]{1}

%\node[] (targettxt) at (\confAngle{-2.5}:2.3) {\vtarget};
%\draw[-latex] (targettxt) to [bend right=20] (r-0) ;

\node[] at (1.6,1.8) {\scriptsize Quasi-axis of symmetry};

\node[] (mr) at (\confAngle{3}:2.8) {\scriptsize leading robot};

\node[] at (\confAngle{1}:0.9) {\scriptsize $r$};
\node[] at (\confAngle{-1}:0.9) {\scriptsize $\bar{r}$};
\node[] at (\confAngle{-2}:1.5) {\scriptsize $r'$};
\node[] at (\confAngle{2}:1.5) {\scriptsize $\bar{r'}$};

\draw[-latex] (mr) to [bend left=20] (r-1) ;

\draw[-{latex}] (\confAngle{-2}:\radius+19) arc (\confAngle{-2}:\confAngle{0.8}:\radius+19);
\node[] at (\confAngle{-0.5}:\radius+25) {\small $d$};

\end{tikzpicture}
\captionof{figure}{An instance of a $\{r,r'\}$-quasi-node-edge symmetric configuration}\label{fig:example-q-symmetric}
        \end{center}
    \end{minipage}
    
\end{figure}

\begin{definition}
Let $C$ be a configuration having a node-edge axis of symmetry $L$. The \emph{target node}, denoted \vtarget, is the only node on $L$, i.e., the only node that is its own symmetric with respect to $L$.
The \emph{main} robots are the two robots that are closest to $\vtarget$, but not on \vtarget. The \emph{secondary} robots are the two robots that are the farthest to \vtarget that have an empty adjacent node (See Figure~\ref{fig:example-symmetric}). 
\end{definition}

%In our algorithm, the goal is to move the main robots toward the target node. If the configuration remains node-edge symmetric, the same process continues until all the robots reaches the target node. If the main robots cannot move due to potential rotational symmetry, then we move the secondary robots as well. If the obtained configuration is not node-edge symmetric, then a robot is crashed and the configuration is quasi-symmetric, as defined below (See Figure~\ref{fig:example-q-symmetric}.). The formal description of our solution is given in Algorithm \ref{algo:ring}.

In our algorithm, the goal is to move the main robots toward the target node. If the configuration remains node-edge symmetric, the same process continues so that all the robots reaches the target node. Nevertheless, in some cases, by moving, the main robots create a periodic configuration which makes the SUIG problem unsolvable. Hence, in this special case, together with the main robots, the secondary robots move as well. 

During this process, if a non node-edge symmetric configuration is reached, then a robot has crashed and the configuration is quasi-symmetric, as defined below (See Figure~\ref{fig:example-q-symmetric}). 

\begin{definition}
Let $C$ be a configuration and $r, r'$ be occupied nodes by robots $r, r'$ in $C$. We say $C$ is $\{r,r'\}$-quasi-node-edge symmetric, with quasi-axis $L$, if $C\setminus\{r\}\cup\{\bar{r'}\}$ is node-edge symmetric of axis $L$, where $\bar{r}\notin C$ (resp. $\bar{r'}\notin C$) is the symmetric of $r$ (resp. $r'$) with respect to $L$ and $r'\in C$ is at distance one from $\bar{r}$.
The target node $v_{target}$ of a quasi-axis is the middle node between $r$ and $\bar{r}$ (there is only one middle node since $n$ is odd). $v_{target}$ is also the only node on $L$.
\end{definition}

From a symmetric configuration, a quasi-axis can be created by our algorithm in two cases: Either the two main robots are ordered to move and one of them is crashed, or the two secondary robots are ordered to move and one of them is crashed. We now define the gap distance associated with a quasi-axis, to distinguish between the two cases.
Recall that, between two points in the ring, since $n$ is odd, there are two paths, one is of odd length and one is of even length. In particular, the target node of a $\{r,r'\}$-quasi-axis of symmetry is located in the middle of the path of odd length between $r$ and $r'$.
\begin{definition}
Let $C$ be $\{r,r'\}$-quasi-node-edge symmetric with target node $v_{target}$. If there is no robot in the path of odd length (between $r$ and $r'$), or if only $v_{target}$ is occupied, then the \emph{gap distance} associated with $\{r,r'\}$ is the length of this path. Otherwise, the \emph{gap distance} associated with $\{r,r'\}$ is the length of the path with even length.
\end{definition}
From this definition, we see that the parity of a quasi-axis determines whether the crashed robot is one of the main robots or not (if not, it is one of the secondary robots, see the algorithm for more details). That is, if the gap distance is odd, the crashed robot is one of the main robots, otherwise, one of the secondary robots.  % (if not, it is one of the farthest \emph{that is able to move}, see the algorithm for more details).

\begin{definition}
Let $C$ be $\{r,r'\}$-quasi-node-edge symmetric with quasi-axis $L$, gap distance $d$, and the path $P$ between $r$ and $r'$ of length $d$. The \emph{leading robot} associated with $\{r,r'\}$ is the robot whose symmetric with respect to $L$ is in $P$. If $r$ is the leading robot, the orientation associated with $\{r,r'\}$ (or with the quasi-axis $L$) is the orientation of the ring such that $r = \rot(d)(r')$ (we also have $\rot(1)(r) = \bar{r'} \notin P$). We call the orientation such that $r' = \rot(d)(r)$ the orientation opposite to $\{r,r'\}$ (or with the quasi-axis $L$).
\end{definition}

Similarly we say that a configuration is \emph{$r$-quasi-periodic} if, by moving $r$ to an adjacent node, the configuration becomes periodic.\\

%Let $P$ be the path that connect $r$ to $r'$ that does not contain other robots. We say the \emph{gap distance $d$} of $\{r,r'\}$ is the length of $P$. The leading robot associated with axis $L$ is the robot in $\{r,r'\}$ such that is symmetric with respect to $L$ is in $P$. 
%If $r$, resp $r'$, is the leading robot, then \emph{direction associated with the axis $L$} is the direction of the ring such that the angle between between $r'$ and $r$, resp. $r$ and $r'$, is the gap distance i.e., $r = \rot(d)(r')$, resp. $r' = \rot(d)(r')$. In other word, if $r$ is the leading robot, then $\rot(1)(r) = \bar{r'} \notin P$.

%We partition the set of configuration into the following subsets:
%\begin{itemize}
%    \item $NE$: the set of node-edge symmetric configurations.
%    \item $Q_{odd}$: the set of quasi-node-edge  symmetric configurations with an even gap distance
%    \item $Q_{even}$: the set of quasi-node-edge  symmetric configurations with an even gap distance
%\end{itemize}

Let $C$ be a possible initial configuration. The idea of the algorithm is to move the main robots towards $v_{target}$. If neither of the two main robots crashes, then the reached configuration remains symmetric and the main robots remain the ones to move. By contrast, if one of the two main robots crashes, then a quasi-node-edge-symmetric configuration $C'$ is reached. Assume without loss of generality that $C'$ is $\{r,r'\}$-quasi-node-edge symmetric with $r$ as the leading robot. The strategy in this case is to move all the robots to their adjacent nodes with respect to the orientation opposite to $\{r,r'\}$. By doing so, the new configuration $C''$ is not only node-edge symmetric but also equivalent to the configuration reached if both main robots have moved. By repeating this process, robots move and eventually reach $v_{target}$. Observe that in some cases, by moving the main robots, a periodic configuration can be reached. To avoid this, robots compute the configuration to be reached when the main robots move. If it is periodic, then in addition to the main robots, the secondary robots move as well. Their destination is their adjacent free node towards $v_{target}$ if it exists, otherwise, they move in the opposite direction of $v_{target}$. \\

Each move of our algorithm is executed on a configuration if the latter satisfies the move's condition. To define the conditions and define formally our solution, we use the following predicates:
\begin{itemize}
    \item $\NE(k)$: is the set of node-edge symmetric configurations having $k$ occupied nodes.
    \item $\QNE(k)$: is the set of quasi-node-edge symmetric configurations having $k$ occupied nodes. We further partition this set depending on the number of quasi-axes and their orientation: 
    \begin{itemize}
    \item $\QNE_{(a,b)}(k)\subset \QNE(k)$ contains $a$ quasi-axes having a given orientation and $b$ quasi-axes with the opposite orientation. 
    \item $\QNE_{(*,b)}(k) = \cup_{a\geq 0} \QNE_{(a,b)}(k)$.
    \end{itemize}
    \item\LFiveFigure{} $\LFive\subset\NE(5)$ (last five): is the set of node-edge symmetric configurations consisting of a block of size five.
    \item\LFourFigure{} $\LFour\subset\NE(4)$ (last four): is the set of node-edge symmetric configurations consisting of four occupied nodes in blocks of size two separated by one empty node.
    \item\LFourPrimeFigure{} $\LFour'\subset\QNE(4)$ (last four prime):  is the set of configurations consisting of four occupied nodes in a block of size three, followed by an empty node and a single occupied node.
    \item\LThreeFigure{} $\LThree\subset\NE(3)$ (last three):  is the set of node-edge symmetric configurations consisting of three occupied nodes, two at distance 4 from each-other and one in the middle.
    \item\LThreeFigureP{} $\LThree'\subset\QNE(3)$ (last three prime):  is the set of configurations consisting of two adjacent occupied nodes, followed by an empty node and a single occupied node.
    \item\LThreeFigurePS{} $\LThree''\subset\NE(3)$ (last three prime-second):  is the set of configurations consisting of a block of size three.
    \item \LTwoFigure{} $\LTwo\subset\NE(2)$ (last two): is the set of node-edge symmetric configurations consisting of two adjacent occupied nodes, and one of those two nodes only hosts crashed robots.
    \item $\Periodic$: is the set of periodic configurations.
%    \item $\OtherConf$: is the set of configurations that are not in $\NE\cup \QNE\cup \Periodic$.
    %\item\QB{I think we should define "last three prime" and "last three second" obtained from "last free" before reaching "last two"} \AL{done}
%    \SK{I added that each set is subset of QNE or NE. We do not use the set ``$O$''. Can we delete it?} 
\end{itemize}

We define three possible moves when the configuration is node-edge symmetric:
\begin{itemize}
    \item \MoveMain: the main robots are ordered to move towards the target node.
    \item \MoveSecondary: the secondary robots move towards the target node if it is empty, or away from the target node otherwise.
    \item \MoveLTwo: all the robots move towards the other occupied node.
\end{itemize}
And two moves when the configuration is not node-edge symmetric, to recover after a crash:
\begin{itemize}
    \item \MoveSame: all the robots move in orientation opposite to the quasi-axis.
    \item \MoveOpposite: the two potentially crashed robots are ordered to move towards their respective target node.
\end{itemize}

The formal description of our solution is given in Algorithm \ref{algo:ring}. 
In the following sections, we show that, when executing Algorithm~\ref{algo:ring} starting from a configuration in $\NE(k)\setminus\Periodic$ ($k>3$), the gathering is achieved by the configuration transitions represented in Figure~\ref{fig:configuration diagram}, regardless of the presence or absence of crashed robots. 

\newcommand{\Reachable}{\ensuremath{\mathcal{C}}\xspace}
%\textcolor{red}{QB: Maybe not used}
Let $\Reachable$ be the set of configurations that are reachable by our algorithm. 
We write $C\oneRound C'$ to denote that $C'$ is obtained from $C\in \Reachable$ after one round of executing Algorithm~\ref{algo:ring}. Similarly $C\overset{*}{\oneRound} C'$ denotes that $C'$ is obtained from $C\in \Reachable$ after a finite number of rounds executing Algorithm~\ref{algo:ring}. To study different cases, we also write $C\oneRoundNoCrash C'$ (resp.  $C\oneRoundMainCrash C'$, or  $C\oneRoundSecCrash C'$) when each robot ordered to move reaches its destination (resp. one of the main robots is crashed, or one of the secondary robots is crashed).

\begin{figure}[t]
    \centering
    \begin{tikzpicture}[scale=0.9]\small
    \node[rectangle, draw] (NEk) at (0, 0) {$\NE(k)\setminus \Periodic$};
    \node[rectangle, draw] (NEkm1) at (0, -1) {$\NE(k-1)\setminus \Periodic$};
    \node[rectangle, draw] (NEkm2) at (0, -2) {$\NE(k-2)\setminus \Periodic$};
    \draw[->] (NEk)  to [out=0, in=90] (2,-1) to [out=270, in=0] (NEkm1);
    \draw[->] (2,-1) to [out=270, in=-10] (NEk);
    \draw[->] (2,-1) to [out=270, in=0] (NEkm2);

    \node[rectangle, draw] (QNEk) at (-3, -1) {$\QNE\setminus \Periodic$};
    \draw[dashed, ->] (NEk) to [out=180, in=40] (QNEk);
    \draw[dashed,->] (QNEk) to [out=0, in=190] (NEk);
    \draw[dashed,->] (QNEk) to [out=0, in=180] (NEkm1);
    \draw[dashed,->] (QNEk) to [out=0, in=180] (NEkm2);

    \node[rectangle, draw] (L3) at (0, -3.2) {$\LThree$};
    \draw[dotted] (0, -2.5) to (0, -2.75);
    \draw[->] (0, -2.75) to (L3);

    \node at (2,0.2) {\scriptsize Lemmas \ref{lem:NE(k) -> NE(k'), without crash}--\ref{lem:NE(k) -> L3, without crash}};
    \node at (-1,-2.7) {\scriptsize Lemma \ref{lem:NE(k) -> L3}};
    \node at (-2.5,-2) {\scriptsize Lemmas \ref{lem:NE(k) -> NE(k'), secondary crash}--\ref{lem:NE(k) -> NE(k'), main crash}};

\begin{comment}[shift={(7,0)}]
    \node[rectangle, draw] (NEk) at (0, 0) {$\NE(5)\setminus \Periodic$};
    \node[rectangle, draw] (NEkm1) at (0, -1) {$\LFive$};
    \draw[->] (NEk)  to [out=0, in=90] (2,-0.5) to [out=270, in=0] (NEkm1);
    \draw[->] (2,-0.5) to [out=270, in=-10] (NEk);

    \node[rectangle, draw] (QNEk) at (-3, -0.5) {$\QNE(5)\setminus \Periodic$};
    \draw[dashed, ->] (NEk) to [out=180, in=40] (QNEk);
    \draw[->] (QNEk) to [out=0, in=190] (NEk);
    \draw[->] (QNEk) to [out=0, in=180] (NEkm1);
\end{comment}

\begin{scope}[shift={(7,0)}]
    \node[rectangle, draw] (NEk) at (0, 0) {$\NE(4)\setminus \Periodic$};
    \node[rectangle, draw] (L4) at (0, -1) {$\LFour$};
    \node[rectangle, draw] (L4p) at (-3, -1.8) {$\LFour'$};
    \node[rectangle, draw] (L3) at (0, -2.2) {$\LThree$};
    \draw[->] (NEk)  to [out=0, in=90] (2,-0.5) to [out=270, in=0] (L4);
    \draw[->] (2,-0.5) to [out=270, in=-10] (NEk);

    \node[rectangle, draw] (QNEk) at (-3, -0.5) {$\QNE(4)\setminus \Periodic$};
    \draw[dashed, ->] (NEk) to [out=180, in=40] (QNEk);
    \draw[dashed, ->] (QNEk) to [out=0, in=190] (NEk);
    \draw[dashed, ->] (QNEk) to [out=0, in=180] (L4);

    \draw[->] (L4) to (L3);
    \draw[dashed, ->] (L4) to[out=200, in=45] (L4p);
    \draw[dashed, ->] (L4p)  to[out=-10, in=180] (L3);

    \node at (2,0.2) {\scriptsize Lemmas \ref{lem:NE(4) -> NE(4) main at distance 2, without crash}--\ref{lem:NE(4) -> L4, with main crash}};
    \node at (-1.5,-1.6) {\scriptsize Lemmas \ref{lem:NE(4) -> NE(4), with secondary crash}--\ref{lem:NE(4) secondary adjacent -> NE(4) non adjacent, with crash}};
    \node at (1,-1.6) {\scriptsize Lemmas \ref{lem:L4 -> L3}};
\end{scope}

\begin{scope}[shift={(3.6,-2.1)}]
    %\node[rectangle, draw] (L5) at (0, -2) {$\LFive$};
    \node[rectangle, draw] (L3) at (-2.5, -1.5) {$\LThree$};
    \node[rectangle, draw] (L3p) at (-1.25, -0.5) {$\LThree'$};
    \node[rectangle, draw] (L3pp) at (0, -1.5) {$\LThree''$};
    %\node[rectangle, draw] (L4) at (0, 0) {$\LFour$};

    \node[rectangle, draw] (L2) at (2, -0.5) {$\LTwo$};
    \node[rectangle, draw] (G) at (2, -1.5) {Gathered};

    %\draw[->] (L5) to (L3);
    %\draw[dashed, ->] (L5) to[out=180, in=90] (L4p);
    \draw[dashed, ->] (L3pp) to [out=10, in=180] (L2);
    \draw[dashed, ->] (L3) to [out=10, in=180] (L3p);
    \draw[dashed, ->] (L3p) to [out=-10, in=170] (L3pp);
    \draw[->] (L3) to [out=-10, in=190] (L3pp);
    \draw[->] (L3pp) to [out=-10, in=190] (G);
    \draw[->] (L2) to (G);

    \node at (0,-2) {\scriptsize Lemma \ref{lem:gather}};
    \node at (3,-1) {\scriptsize Lemma \ref{lem:L2}};
\end{scope}
\end{tikzpicture}
    \caption{Overview of the transitions between configurations. Dashed lines represent transitions that occur only when a crashed robot is ordered to move. Loops are shown to occur only a finite number of times, see the corresponding lemmas for more details.}
    \label{fig:configuration diagram}
\end{figure}
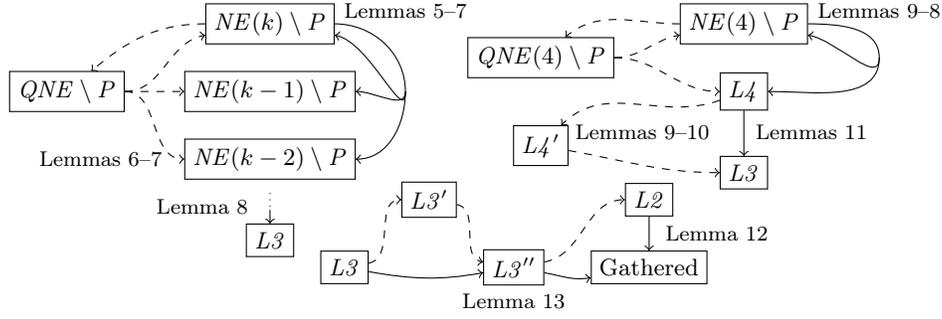

\begin{comment}
\textcolor{red}{QB: tried a table but it is not very readable.}\\
\begin{tabular}{rcl}
   \hline& $C\overset{\text{no crash}}{\oneRound} C'$&\\\hline\hline
    $C\in \LTwo$ & \MoveLTwo & Gathered \\\hline\hline
    &$C\in \NE(k)$ &\\\hline
    $k \in\{4,5\}$ and $C\notin \LFour\cup \LThree$ &\MoveSecondary & eventually $C'\in \LFour\cup \LThree$\\\hline
    almost periodic ? &\MoveMain and \MoveSecondary & $\NE(k')\setminus\Periodic$,  $k'\leq k$\\
   otherwise & \MoveMain &  and $|C'|\leq 3 \Rightarrow C'\in \LThree$ \\
   \hline
   \hline
    &$C\in \QNE(k)$ &\\\hline
   $C\in\QNE_{(*,0)}(k)$ & \MoveSame & $C'\in \NE\setminus\Periodic$\\
   \hline
   $C\in\QNE_{(1,1)}(k)$ & \MoveOpposite & $C'\in \NE\setminus\Periodic$\\
   \hline
\end{tabular}
\end{comment}

%
% Using minipageFigure command:
% \minipageFigure{percentage width of the page}{3 nodes}{radius}
% {caption}{label}
% {other draw command}
%  
\begin{algorithm}[t]
\uIf{$C\in \LTwo$}{
\MoveLTwo %: Move towards the adjacent occupied node
}
\uElseIf{$C\in \NE(k)$} {
    \uIf{$k \in \{4, 5\}$, $(n,k)\neq (7,4)$, $C\notin \LFour\cup\LFive$, and the main robots are adjacent to $v_{\mathit{target}}$\label{algo:45}}{
        \MoveSecondary (See Figure~\ref{fig:algo-sym-sec-to-target})
    }
    \uElseIf{Moving the main robots toward the target node does not create periodic configuration\label{algo:ring:line check not periodic}} {
        \MoveMain (See Figure~\ref{fig:algo-sym-main-to-target})
    }
    \uElse{
        \MoveMain and \MoveSecondary (See Figures~\ref{fig:algo-sym-ms-to-target} and~\ref{fig:algo-sym-ms-away-target}
        %, and~ \ref{fig:algo-sym-ms-to-target-2}
        ) \label{algo:ring:periodic}
    }
}
\uElseIf{$C\in\QNE_{(*,0)}$\label{alg:all}}{
\MoveSame (See Figure~\ref{fig:algo-quasi-sym})
}
%\textcolor{red}{\uElseIf{The configuration is quasi-node-edge-symmetric with an even gap distance}{??}}
\uElseIf{$C\in\QNE_{(1,1)}$}{
\MoveOpposite (See Figure~\ref{fig:algo-two opposite quasi axis})
}
\caption{SUIG algorithm on ring-shaped network}\label{algo:ring}
\end{algorithm}

\begin{figure}[t]
\minipageFigure{0.47}{21}{1.5cm}
{Only the secondary robots are ordered to move towards the target node.}
{\label{fig:algo-sym-sec-to-target}}
{
\confAxis{0}{-4}{20}
\confRobot[green]{-9}
\confRobot[green]{9}

\confRobot[]{-1}
\confRobot[]{1}
\confRobot[]{0}

\draw[-{latex},thick] (r-9) -- (r-8);
\draw[-{latex},thick] (r-12) -- (r-13);

\node[] (targettxt) at (\confAngle{-2.5}:2.2) {\vtarget};

\node[] at (1.3,1.8) {\scriptsize Axis of symmetry};

\draw[-latex] (targettxt) to [bend right=15] (r-0) ;
}~~~~
\minipageFigure{0.47}{21}{1.5cm}
{The main robots are ordered to move towards the target node.}
{\label{fig:algo-sym-main-to-target}}
{
\confAxis{0}{-4}{20}
\confRobot{-10}
\confRobot{10}
\confRobot[]{-9}
\confRobot[]{9}
\confRobot{-6}
\confRobot{6}
\confRobot[green]{-2}
\confRobot[green]{2}

\draw[-{latex},thick] (r-2) -- (r-1);
\draw[-{latex},thick] (r-19) -- (r-20);

\node[] (targettxt) at (\confAngle{-2.5}:2.2) {\vtarget};
\node[] at (1.3,1.8) {\scriptsize Axis of symmetry};
\draw[-latex] (targettxt) to [bend right=15] (r-0) ;
}
\end{figure}

\begin{figure}[t]
\minipageFigure{0.45}{21}{1.5cm}
{The main robots and the secondary robots (to avoid reaching a periodic configuration) are ordered to move towards the target node.}
{\label{fig:algo-sym-ms-to-target}}
{
\confAxis{0}{-4}{20}
\confRobot[blue]{-8}
\confRobot[blue]{8}
\confRobot{-6}
\confRobot{6}
\confRobot[green]{-2}
\confRobot[green]{2}

\draw[-{latex},thick] (r-2) -- (r-1);
\draw[-{latex},thick] (r-19) -- (r-20);
\draw[-{latex},thick] (r-8) -- (r-7);
\draw[-{latex},thick] (r-13) -- (r-14);

\node[] (targettxt) at (\confAngle{-2.5}:2.2) {\vtarget};
\node[] at (1.3,1.8) {\scriptsize Axis of symmetry};
\draw[-latex] (targettxt) to [bend right=15] (r-0) ;
}~~~~
\minipageFigure{0.5}{33}{2.0cm}
{The main robots are ordered to move towards the target node and, to avoid reaching a periodic configuration, the secondary robots move away from it (it is their only empty adjacent node).}
{\label{fig:algo-sym-ms-away-target}}
{
\confAxis{0}{-4}{20}

\confRobot[green]{-2}
\confRobot[green]{2}
\confRobot[]{-3}
\confRobot[]{3}
\confRobot[]{-4}
\confRobot[]{4}
\confRobot[]{-7}
\confRobot[]{7}
\confRobot[]{-8}
\confRobot[]{8}
\confRobot[]{-10}
\confRobot[]{10}
\confRobot[]{-12}
\confRobot[]{12}
\confRobot[]{-14}
\confRobot[]{14}
\confRobot[blue]{-15}
\confRobot[blue]{15}

\draw[-{latex},thick] (r-2) -- (r-1);
\draw[-{latex},thick] (r-31) -- (r-32);
\draw[-{latex},thick] (r-15) -- (r-16);
\draw[-{latex},thick] (r-18) -- (r-17);

\node[] (targettxt) at (\confAngle{-2.5}:2.7) {\vtarget};
\node[] at (1.3,2.4) {\scriptsize Axis of symmetry};
\draw[-latex] (targettxt) to [bend right=15] (r-0) ;
}
\begin{comment}
\minipageFigure{0.43}{27}{2cm}
{The main and secondary robots are ordered to move towards the target node. Indeed, moving only the main robots would result in a periodic configuration.\SK{Is this figure needed? The case is on Figure 7.}}
{\label{fig:algo-sym-ms-to-target-2}}
{
\confAxis{0}{-4}{20}
\confRobot{-13}
\confRobot{13}
\confRobot[blue]{-12}
\confRobot[blue]{12}
\confRobot{-10}
\confRobot{10}
\confRobot{-8}
\confRobot{8}
\confRobot{-6}
\confRobot{6}
\confRobot{-5}
\confRobot{5}
\confRobot{-4}
\confRobot{4}
\confRobot{-3}
\confRobot{3}
\confRobot[green]{-2}
\confRobot[green]{2}

\draw[-{latex},thick] (r-2) -- (r-1);
\draw[-{latex},thick] (r-25) -- (r-26);
\draw[-{latex},thick] (r-15) -- (r-16);
\draw[-{latex},thick] (r-12) -- (r-11);

\node[] (targettxt) at (\confAngle{-2.5}:2.6) {\vtarget};
\node[] at (1.4,2.4) {\scriptsize Axis of symmetry};
\draw[-latex] (targettxt) to [bend right=15] (r-0) ;
}
\end{comment}
\end{figure}

\begin{figure}[h]
\minipageFigure{0.51}{21}{1.5cm}
{When all the quasi-axis have the same orientation. All the robots are ordered to move (one of them is crashed).}
{\label{fig:algo-quasi-sym}}{
\confAxis{0}{-4}{20}

\confRobot{-10}
\confRobot{10}
\confRobot{-9}
\confRobot{9}
\confRobot{-6}
\confRobot{6}
\confRobot[label={[label distance=-0.05cm]45:{\scriptsize $r_1'$}}]{-2}
\confRobot[label={[label distance=-0.05cm]90:{\scriptsize $r_1$}}]{1}

\node[] at (1.7,1.85) {\scriptsize Quasi-axis of symmetry $L_1$};

\draw[-{latex}] (\confAngle{-2}:\radius+19) arc (\confAngle{-2}:\confAngle{0.8}:\radius+19);
\node[] at (\confAngle{-0.5}:\radius+15) {\small $d_1$};

\draw[-{latex},thick] (r-1) -- (r-0);
\draw[-{latex},thick] (r-6) -- (r-5);
\draw[-{latex},thick] (r-9) -- (r-8);
\draw[-{latex},thick] (r-10) -- (r-9);
\draw[-{latex},thick] (r-11) -- (r-10);
\draw[-{latex},thick] (r-12) -- (r-11);
\draw[-{latex},thick] (r-15) -- (r-14);
\draw[-{latex},thick] (r-19) -- (r-18);
}
\minipageFigure{0.47}{31}{1.7cm}
{When there are two quasi-axes with opposite orientations. The two potentially crashed robots (orange robots) are ordered to move (one of them is crashed).}
{\label{fig:algo-two opposite quasi axis}}{
\confAxis{0}{-4}{20}
\confAxis{8}{-4}{20}

\confRobot[blue]{-1}
\confRobot[orange]{2}
\confRobot[blue]{14}
\confRobot[blue]{17}

\confRobot[orange]{5}
\confRobot[blue]{10}
\confRobot[blue]{-5}
\confRobot[blue]{-10}
\draw[-{latex},thick] (r-2) -- (r-1);
\draw[-{latex},thick] (r-5) -- (r-6);

\node[] at (1.7,2.05) {\scriptsize Quasi-axis of symmetry $L_1$};
\node[] at (-2.1,0.1) {\scriptsize $L_2$};

\draw[-{latex}] (\confAngle{2}:\radius+10) arc (\confAngle{2}:\confAngle{-0.8}:\radius+10);
\node[] at (\confAngle{0.5}:\radius+20) {\small $d_1$};

\draw[-{latex}] (\confAngle{5}:\radius+10) arc (\confAngle{5}:\confAngle{9}:\radius+10);
\node[] at (\confAngle{7}:\radius+20) {\small $d_2$};
}
\end{figure}

\renewcommand{\bar}[1]{\overline{#1}}
\newcommand\apps[1]{\stackrel{\mathclap{\normalfont\mbox{$s_{#1}$}}}{\longrightarrow}}
\newcommand\apprho[1]{\stackrel{\mathclap{\normalfont\mbox{$\rho^{#1}$}}}{\longrightarrow}}

%%%%%%%%%%%%%%%%%%%%%%%%%%%%%%%%%%%%%%%%%%%%%%%%%%%
%
%      Overview of the proof of correctness
%
%%%%%%%%%%%%%%%%%%%%%%%%%%%%%%%%%%%%%%%%%%%%%%%%%
\subsection{Overview of the proof of correctness}
To prove the correctness of our algorithm, we first show that from a configuration with $k>4$ occupied nodes, we reach in finite time a configuration with $k-1$ or $k-2$ occupied nodes (Subsect.~\ref{sec:NE->L3}). 
Repeating the argument, we eventually reach a configuration with either four or three occupied locations.
Then, from three occupied locations, we show that after at most four rounds, the gathering is complete (Subsect.~\ref{sec:L3->Gathered}).
When starting from four occupied nodes, the proof is different, but the result is the same (Subsect.~\ref{sec:NE(4)->L3}). 

During the execution, when the configuration $C$ is node-edge symmetric, either 2 or 4 robots are ordered to move. If the moving robots are not crashed, the configuration remains node-edge-symmetric and the main robots eventually reach the target node. When one of these robots is crashed, the obtained configuration $C'$ may not be node-edge symmetric anymore. In fact, this is the easiest case to consider because it implies that the configuration is quasi-node-edge symmetric, so all the robots can somehow realize that a crash occurred, and the algorithm can recover. Indeed, if all quasi-axes have the same orientation, then by moving all robots in the opposite orientation of that of the quasi-axis, the obtained configuration $C''$ is similar to the configuration obtained by moving the crashed robot towards its destination. So, the obtained configuration $C''$ is similar to the configuration obtained when no robot is crashed, i.e., $C\oneRoundNoCrash C''$. If $C'$ contains two quasi-axes with different orientation, then executing the move $\MoveOpposite$ ensures that the obtained configuration is node-edge symmetric with a new axis of symmetry such that the main robots are correct.
The last remaining case is when, after a crash, the configuration remains node-edge symmetric, so the robots do not realize that a crash occurred. In this case, we prove that the main robots in $C'$ are correct, and we can apply a previous case.
%\textcolor{red}{
Due to space limitation, some proofs of the following sections are presented in Appendix.%}

%%%%%%%%%%%%%%%%%%%%%%%%%%%%%%%%%%%%%%%%%%%%%%%%%%%
%
%      Geometric Results
%
%%%%%%%%%%%%%%%%%%%%%%%%%%%%%%%%%%%%%%%%%%%%%%%%%
\subsection{Geometric Results}\label{sec:geometric results}

\begin{theorem}[disc]{lemma}\label{lem:periodic configurations are determined by k/3 consecutive positions}
Let $C$ and $C'$ be two periodic configurations with $k$ occupied positions, and where more than $n/3$ consecutive nodes have the same state (occupied or empty) in both configurations, then $C=C'$.
%\SK{the consecutive occupid nodes may be a sequence of 0,1 (empty and occupied). Is not it correct?}
\end{theorem}
\begin{proof}
Since $n$ is odd, the period of each symmetry is at least $n/3$. So the entire configuration is determined by the state of any $n/3$ consecutive nodes. This implies that $C=C'$.\qed
\end{proof}

\begin{theorem}[disc]{lemma}\label{lem:2 axis of sym => bisector is an axis}
%Let the number of nodes $n$ be odd. AL: as the initial conf is node-edge sym, we know that n is odd 
If a configuration $C$ has two axes of symmetry, then the unique bisector line of these axes passing through a node is also an axis of symmetry.
\end{theorem}
\begin{proof}
Let $q$ be the number of axes of symmetry. Since $n$ is odd, $q$ must be odd as well. Let $L_i$, with $0\leq i < q$ be the axis of symmetry of $C$. We can order the nodes of the ring $u_0, u_1, \ldots u_{n-1}$ so that $L_0$ contains $u_0$ and, more generally $L_i$ contains $u_{in/q}$. Now we show that the bisector line of any pair of axes $L_i$ and $L_j$, with $i\neq j$, is another axis.

If $i+j$ is even, then, clearly, $L_{\frac{i+j}{2}}$ is the line bisector of $L_i$ and $L_j$. Otherwise, one can see that the line bisector is $L_{q'}$, with $q'\equiv \frac{i+j+q}{2} \mod q$.\qed
\end{proof}

We know that if a secondary robot that is ordered to move is crashed, the resulting configuration is quasi-node-edge-symmetric and quasi-periodic. 
The next lemma states that in this case there exists no other quasi-axis of symmetry.
\begin{theorem}[disc]{lemma}\label{lem:even quasi-axis is unique}
If a configuration with more than 3 occupied nodes has a $\{r_1, r_1'\}$-quasi-axis of symmetry with even gap distance, and is $r_1$-quasi-periodic, then it does not have another quasi-axis (regardless of its parity).
\end{theorem}
\begin{proof}
Let $C$ be a configuration that is $\{r_1, r_1'\}$-quasi-symmetric with an even gap distance. We know that this also implies that $C\setminus\{r_1\}\cup\{\bar{r_1'}\}$ is periodic with period $p\geq 3$ (that is odd, because $n$ is odd) and is symmetric, with axis $L_1$, so the string of node states is $S=(a_{m}a_{m-1}\ldots a_2a_1a_0 a_1a_2\ldots a_{m-1} a_m)^{n/p}$, with $m = (p-1)/2$. Assume for the sake of contradiction that there is another $\{r_2, r_2'\}$-quasi-axis $L_2$. The gap distance associated with $L_2$ is smaller than $n/3$, otherwise, let $P$ be the path joining $r_2$ and $r_2'$ of length the gap distance. Either, (i) there is no other robot on $P$, so there are at most three occupied nodes (recall that by moving one robot the configuration is periodic), or (ii) there is one robot on $P$, or (iii) $P$ consists only of occupied nodes so that the whole configuration as well. Each case implies a contradiction.
So there are two paths $P_1$ and $P_2$ that are disjoint, symmetric with respect to $L_2$ and have length $p$. $r_1$ and $\bar{r_1}'$ cannot be on different sides of $L_2$ (because they have different state so they cannot be symmetric -- recall that they are adjacent) so lets say $P_1$ contains $r_1$ and $\bar{r_1}'$.
Since $P_1$, resp. $P_2$, is of length at least $p$, it contains a full copy of the periodic string of states $S$ (up to cyclic permutation). 
%Let $i_1$ the index in the string that correspond to the state of the node at $r_1$. 
Moreover, by assumption, $r_1$ should not be there for the configuration to be periodic so that, in the string of state in $P_1$ is actually cyclic  permutation $\sigma_1$ of $S'= a_m a_{m-1}\ldots \lnot a_{i_1+1} \lnot a_{i_1} \ldots a_2a_1a_0 a_1\ldots a_{m-1} a_m$ for some $i_1\in [0, m]$, while in $P_2$ it is a cyclic permutation $\sigma_2$ of $S$. So we have $\sigma_1S' = \sigma_2S$ (because $P_1$ and $P_2$ are symmetric with respect to $L_2$). This implies that $\lnot a_{i_1} = a_{i_2}$ and also that $a_{i_2} = a_{i_1}$, which is a contradiction. \qed
% (recall that $a_{i_1}$ appears two times in the string $S$), which is a contradiction. If $i_1 = 0$, then $\lnot a_{0} = a_{i_2}$ but we also have $a_{0} = a_{i_2}$
%Maybe I should rewrite this with better notation, maybe by using indexes of nodes instead of states, and with a function returning the state, idk.
\end{proof}

\vspace{-2mm}
\begin{theorem}[disc]{lemma}
\label{lem:initial config remains non-periodic with crash}
Let $C\in \NE(k) \setminus P$ with $k>3$ occupied nodes, and $C\oneRound C'$, then $C'\not\in\Periodic$, even if there exists a crashed robot.
\end{theorem}
\begin{proof}
The proof is similar to the proof of Lemma~\ref{lem:even quasi-axis is unique}. Indeed, $C'$ is quasi-periodic and has an axis of symmetry. \qed
\end{proof}
\vspace{-2mm}
\begin{theorem}[disc]{lemma}
Let $C\in \QNE(k)\setminus \NE(k) \setminus P$ with $k>3$. If all quasi-axes of $C$ have odd gap distances, then, either
\begin{itemize}
    \item all quasi-axes have the same orientation, or
    \item there are exactly two quasi-axes with opposite orientation.
\end{itemize}
\end{theorem}
%\begin{lemma}
%If a configuration $C$ is $\{r_1, r_1'\}$-quasi-node-edge symmetric and $\{r_2, r_2'\}$-quasi-node-edge symmetric, with $\{r_1, r_1'\}\neq \{r_2, r_2'\}$, with gap distances having the same parity, then $C$ is node-edge symmetric or the direction associated with each quasi-axis is the same.
%\end{lemma}
\begin{proof}
Consider two arbitrary quasi-axes of $C$, $\{r_1, r_1'\}$-quasi-axis $L_1$ and  $\{r_2, r_2'\}$-quasi-axis $L_2$. 
Recall the notations such that $C\setminus\{r_1\}\cup\{\bar{r_1'}\}$ (resp. $C\setminus\{r_2\}\cup\{\bar{r_2'}\}$) are node-edge symmetric with respect to $L_1$ (resp. $L_2$). In particular, we have that $L_1\neq L_2$. 

Let $d_1$ be the gap distance of $\{r_1, r_1'\}$ and $d_2$ be the gap distance of $\{r_2, r_2'\}$.
By assumption, $d_1$ and $d_2$ are odd.

Without loss of generality, we assume $r_1$ is the leading robot among $\{r_1, r_1'\}$.
From now, we choose for the orientation of the rotations the orientation associated $\{r_1, r_1'\}$, i.e., such that $r_1 = \rot(d_1)(r_1')$ (in the following figures, we use the counterclockwise direction).

Let $s_1$ (resp. $s_2$) be the symmetry transformation with respect to $L_1$ (resp. $L_2$). For any nodes $u$ and $v$, we write $u \apps{1} v$ when $v = s_1(u)$.
Consider the following longest sequence of occupied nodes $u_0, u_1, u_2, \ldots, u_m$ such that
\begin{equation}\label{eq:seq from r1 to r2}
r_1 = u_0 \apps{2} u_1 \apps{1} u_2\apps{2} \ldots \apps{i} u_m
\end{equation}
We have to consider two cases depending on the value of $i$, either the last symmetry is $s_1$ or $s_2$.\\
\textbf{Case 1: $i=1$.}
In this case, $s_2(u_m)$ is not a robot (otherwise the sequence is not maximal). The only robots that do not have symmetric robots with respect to $L_2$ are $r_2$ and $r_2'$, so $u_m$ is either $r_2$ or $r_2'$. Without loss of generality (since both robots are interchangeable in the definition) assume $u_m = r_2$.

Similarly, there exists another sequence 
\[
r_1' \apps{2} u_1' \apps{1} u_2'\apps{2} \ldots \apps{i} u'_{m'}
\]
where $u'_{m'} = r_2'$.

We consider four sub-cases (recall that we assumed w.l.o.g. that $r_1 = \rot(d_1)(r_1')$, otherwise the results is obtained by symmetry).

\begin{minipage}{0.49\textwidth}
\centering

\begin{tikzpicture}
  \tikzstyle{robot}=[circle,fill=black!100,inner sep=0.03cm]
\def \n {41}
\def \radius {2cm}
\def \margin {8} % margin in angles, depends on the radius

\foreach \s in {0,...,\n}
{
  \node[draw, circle,inner sep=0.03cm] at ({360/\n * \s}:\radius) {};
}
\draw ({360/\n * 29}:\radius+10) --  ({360/\n * (\n/2 + 29)}:\radius+10);
\draw ({360/\n * 1}:\radius+10) --  ({360/\n * (\n/2 + 1)}:\radius+10);

\node[robot,label={-90:$r_1$}]  at ({360/\n * 31}:\radius) {};
\node[robot,label={-135:$r_1'$}]  at ({360/\n * 26}:\radius) {};

\node[robot,label={0:$r_2$}]  at ({360/\n * (-2)}:\radius) {};
\node[robot,label={0:$r_2'$}]  at ({360/\n * 3}:\radius) {};

\draw[-{latex}] ({360/\n * (26)}:\radius+19) arc ({360/\n * (26)}:{360/\n * (30.9)}:\radius+19);
\draw[-{latex}] ({360/\n * (-1.9)}:\radius+18) arc ({360/\n * (-1.9)}:{360/\n * (3)}:\radius+18);

\node[] at ({360/\n * (30)}:\radius+25) {\small $d_1$};
\node[] at ({360/\n * (-0.5)}:\radius+28) {\small $d_2$};

\end{tikzpicture}
    
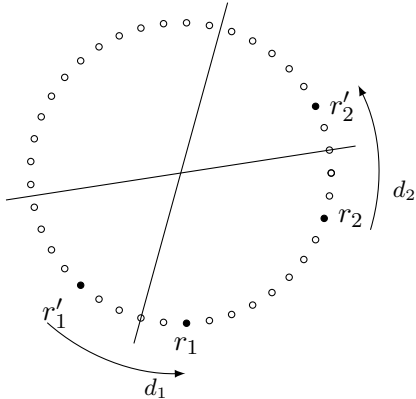
\captionof{figure}{Case 1.a}
    \label{fig:case 1.a}
\end{minipage}
\begin{minipage}{0.49\textwidth}
    \centering

\begin{tikzpicture}
  \tikzstyle{robot}=[circle,fill=black!100,inner sep=0.03cm]
\def \n {41}
\def \radius {2cm}
\def \margin {8} % margin in angles, depends on the radius

\foreach \s in {0,...,\n}
{
  \node[draw, circle,inner sep=0.03cm] at ({360/\n * \s}:\radius) {};
}
\draw ({360/\n * 29}:\radius+10) --  ({360/\n * (\n/2 + 29)}:\radius+10);
\draw ({360/\n * 1}:\radius+10) --  ({360/\n * (\n/2 + 1)}:\radius+10);

\node[robot,label={-90:$r_1$}]  at ({360/\n * 31}:\radius) {};
\node[robot,label={-135:$r_1'$}]  at ({360/\n * 26}:\radius) {};

\node[robot,label={0:$r_2'$}]  at ({360/\n * (-2)}:\radius) {};
\node[robot,label={0:$r_2$}]  at ({360/\n * 3}:\radius) {};

\draw[-{latex}] ({360/\n * (26)}:\radius+19) arc ({360/\n * (26)}:{360/\n * (30.9)}:\radius+19);
\draw[-{latex}] ({360/\n * (-1.9)}:\radius+18) arc ({360/\n * (-1.9)}:{360/\n * (3)}:\radius+18);

\node[] at ({360/\n * (30)}:\radius+25) {\small $d_1$};
\node[] at ({360/\n * (-0.5)}:\radius+28) {\small $d_2$};

\end{tikzpicture}
    
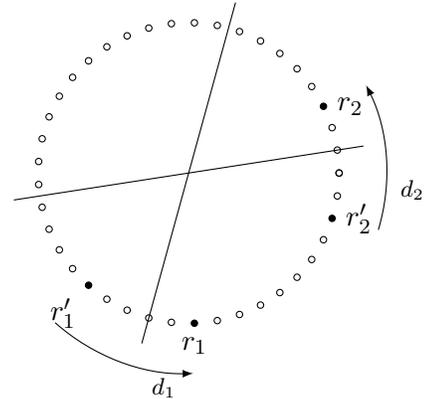
\captionof{figure}{Case 1.b}
    \label{fig:case 1.b}
\end{minipage}

\textbf{Case $1.a$:} $r_2' = \rot(d_2)(r_2)$ and $r_2'$ is the leading robot (see Figure~\ref{fig:case 1.a}).

In this case, the orientation of the two quasi-axes are the same.

\textbf{Case $1.b$:} $r_2 = \rot(d_2)(r_2')$ and $r_2$ is the leading robot (see Figure~\ref{fig:case 1.b}).%; or $r_1 = \rot(x+1)(l_1)$ and $r_2 = \rot(x+1)(l_2)$.

In this case, the orientation of the two quasi-axis are the same.

%
%    CASE 1.c
% 
\textbf{Case $1.c$:} $r_2 = \rot(d_2)(r_2')$ and $r_2'$ is the leading robot (see Figure~\ref{fig:case 1.c}).

In this case, we need a more careful analyze of the situation. We know that applying two axial symmetries is equivalent to applying a rotation whose angle is two times the angle formed by the two axes of symmetry (in more details, the angle of the rotation associated with the composition of axial symmetry of axis $L_2$ followed by the axial symmetry of axis $L_1$ is two times the angle of the rotation mapping $L_2$ to $L_1$).
Here, let $\rho$ be the angle of the rotation $s_1\circ s_2$. In other words, $\rho$ is such that $\rot(\rho/2)(l_2)=l_1$, where $l_1$ (resp. $l_2$) is the target node associated with $L_1$ (resp. $L_2$).

\begin{minipage}{0.49\textwidth}
    \centering
\begin{tikzpicture}
  \tikzstyle{robot}=[circle,fill=black!100,inner sep=0.03cm]
\def \n {41}
\def \radius {2cm}
\def \margin {8} % margin in angles, depends on the radius

\foreach \s in {0,...,\n}
{
  \node[draw, circle,inner sep=0.03cm] at ({360/\n * \s}:\radius) {};
}

\draw[red] ({360/\n * 29}:\radius+10) --  ({360/\n * (\n/2 + 29)}:\radius+10);
\draw[blue] ({360/\n * 1}:\radius+10) --  ({360/\n * (\n/2 + 1)}:\radius+10);

\node[red] at (-119:1) {\scriptsize $L_1$};
\node[blue] at (20:1.4) {\scriptsize $L_2$};

\draw ({360/\n * 15}:\radius+10) --  ({360/\n * (\n/2 + 15)}:\radius+10);
\node[] at (-40:1.2) {\scriptsize $L_3$};

\node[robot,label={-90:$r_1$}]  at ({360/\n * 31}:\radius) {};
\node[robot,label={[label distance=-0.1cm]45:{ $r_1'$}}]  at ({360/\n * 26}:\radius) {};

\node[robot,label={0:$r_2'$}]  at ({360/\n * (-1)}:\radius) {};
\node[robot,label={0:$r_2$}]  at ({360/\n * 4}:\radius) {};

\draw[-{latex}] ({360/\n * (26)}:\radius+15) arc ({360/\n * (26)}:{360/\n * (30.9)}:\radius+15);
\node[] at ({360/\n * (30)}:\radius+25) {\small $d_1$};

\draw[-{latex}] ({360/\n * (31)}:\radius+15) arc ({360/\n * (31)}:{360/\n * (\n+4)}:\radius+15);
\node[] at ({360/\n * (35)}:\radius+25) {\small $\rho m/2$};

\draw[-{latex}] ({360/\n * (-1)}:\radius+25) arc ({360/\n * (-1)}:{360/\n * (4)}:\radius+25);
\node[] at ({360/\n * (1.5)}:\radius+40) {\small $d_2$};

%\draw[-{latex}] ({360/\n * (29)}:\radius+35) arc ({360/\n * (29)}:{360/\n * (\n+1)}:\radius+35);
%\node[] at ({360/\n * (35)}:\radius+45) {\small $-\rho/2$};

\end{tikzpicture}
    \captionof{figure}{Case 1.c}
    \label{fig:case 1.c}
\end{minipage}
~
\begin{minipage}{0.49\textwidth}
\centering
\begin{tikzpicture}
  \tikzstyle{robot}=[circle,fill=black!100,inner sep=0.03cm]
\def \n {41}
\def \radius {2cm}
\def \margin {8} % margin in angles, depends on the radius

\foreach \s in {0,...,\n}
{
  \node[draw, circle,inner sep=0.03cm] at ({360/\n * \s}:\radius) {};
}
\draw[red] ({360/\n * 29}:\radius+10) --  ({360/\n * (\n/2 + 29)}:\radius+10);
\draw[blue] ({360/\n * 1}:\radius+10) --  ({360/\n * (\n/2 + 1)}:\radius+10);

\node[robot,label={-90:$r_1$}]  at ({360/\n * 31}:\radius) {};
\node[robot,label={-135:$r_1'$}]  at ({360/\n * 26}:\radius) {};

\node[robot,label={0:$r_2$}]  at ({360/\n * (-1)}:\radius) {};
\node[robot,label={0:$r_2'$}]  at ({360/\n * 4}:\radius) {};

\draw[-{latex}] ({360/\n * (26)}:\radius+19) arc ({360/\n * (26)}:{360/\n * (30.9)}:\radius+19);
\draw[-{latex}] ({360/\n * (-0.9)}:\radius+17) arc ({360/\n * (-0.9)}:{360/\n * (3.9)}:\radius+17);

\node[] at ({360/\n * (29)}:\radius+25) {\small $d_1$};
\node[] at ({360/\n * (1)}:\radius+25) {\small $d_2$};

\draw[dashed,blue] ({360/\n * 31}:\radius) --  ({360/\n * (12)}:\radius);
\draw[dashed,blue] ({360/\n * 26}:\radius) --  ({360/\n * (17)}:\radius);
\node[red] at (-115:1.5) {\scriptsize $L_1$};
\node[blue] at (20:1.4) {\scriptsize $L_2$};

\draw ({360/\n * 15}:\radius+10) --  ({360/\n * (\n/2 + 15)}:\radius+10);
\node[] at (-40:1.2) {\scriptsize $L_3$};

\draw[dashed,red] ({360/\n * 4}:\radius) --  ({360/\n * (13)}:\radius);
\draw[dashed,red] ({360/\n * (-1}:\radius) --  ({360/\n * (18)}:\radius);

\node[robot,label={135:{\scriptsize $s_1(r_2)$}}]  at ({360/\n * 13}:\radius) {};
\node[robot,label={135:{\scriptsize $s_2(r_1')$}}]  at ({360/\n * 17}:\radius) {};
\node[robot,label={90:{\scriptsize $s_2(r_1)$}}]  at ({360/\n * 12}:\radius) {};
\node[robot,label={180:{\scriptsize $s_1(r_2')$}}]  at ({360/\n * 18}:\radius) {};

\end{tikzpicture}
    \captionof{figure}{Case 1.d}
    \label{fig:case 1.d}
\end{minipage}

From (\ref{eq:seq from r1 to r2}) we deduce that $r_2 = \rot(\rho m/2)(r_1)$. Similarly 
\begin{equation}\label{eq:case1c eq1}
    r_2' = \rot(\rho m'/2)(r_1').
\end{equation}
Moreover, in this case we have 
\begin{align*}
\left.\begin{array}{rl}
r_1 &= \rot(d_1)(r_1')\\
r_2 &= \rot(\rho m/2)(r_1)\\
r_2' &= \rot(-d_2)(r_2)\\
\end{array}\right\}
\Rightarrow r_2' = %\rot(-2x-1 + \rho m/2 + 2x+1)(r_1') =
\rot(\rho m/2+d_1-d_2)(r_1').
\end{align*}
With Equation~(\ref{eq:case1c eq1}), this implies
\[
\rot\left(\rho\frac{m}{2}+d_1-d_2\right) = \rot\left(\rho\frac{m'}{2}\right)
\]
so that 
\begin{equation}\label{eq:case1c:eq2}
\rot(\rho\frac{m - m'}{2}) = \rot(d_2-d_1).
\end{equation}

Another interesting observation is to see that, since $\overline{r_1} = s_1(r_1)$ and $\overline{r_2} = s_2(r_2)$, we in fact have $\overline{r_2} = \rot\left(-\rho\left(\frac{m}{2} + 1\right)\right)(\overline{r_1})$ (it is like extending the sequence $(u_i)_i$ one point before and one after, and the sequence of symmetries now starts with $s_1$ instead of $s_2$, so the rotation is opposite). And since $\overline{r_1} = \rot(1)(r_1')$ and $\overline{r_2} = \rot(1)(r_2')$, then we have 
\[
r_2' = \rot\left( -\rho\left(\frac{m}{2} + 1\right) + 2\right)(r_1')
\]
with Equation~(\ref{eq:case1c eq1}) we obtain
\begin{equation}\label{eq:case1c:(m+m')/2+1}
\rot\left(\rho\left(\frac{m+m'}{2} + 1\right)\right)  = \rot\left(2\right)
\end{equation}

We now show that $d_2 = d_1$ and $m=m'$. First, assume for the sake of contradiction that $d_1 \neq d_2$ (then clearly $m\neq m'$, otherwise $\rot\left(\rho\frac{m - m'}{2}\right) = \rot(0) = \rot(d_2-d_1)$ is a contradiction).
Assume that $d_1 < d_2$ and $m>m'$ (the other cases are proved in a similar way).
If $(m-m')/2$ is even, Equation~(\ref{eq:case1c:eq2}) implies $\rot(\rho\frac{m - m'}{4}) = \rot\left(\frac{d_2-d_1}{2}\right)$. So we have
\[
u_{m-2\frac{m-m'}{4}} = \rot\left(-\rho\left(\frac{m-m'}{4}\right)\right)(r_2)= \rot\left(-\frac{d_2-d_1}{2}\right)(r_2).
\]
so the robot at $u_{m-2\frac{m-m'}{4}}$ is between $r_2'$ and $r_2$, at distance $0 < \frac{d_2-d_1}{2} < d_2$ from $r_2$. But a robot can be between $r_2$ and $r_2'$ only if it is on the quasi-axis $L_2$ (it is the target node), so at distance $(d_2 + 1)/2$ from $r_2$, which implies that $d_1= - 1$, which is a contradiction. 
%If $(m-m')/2$ is odd, so is  $(m+m')/2$ (because $m'$ is even), so $(m+m')/2 + 1$ is even. 
If $(m-m')/2$ is odd, then $(m+m')/2$ is also odd (because $(m+m')/2 = (m-m')/2+m'$ and $m'$ is even). Thus, $(m+m')/2 + 1$ is even.
However, we know that $\rot\left(-\rho\left(\frac{m+m'}{2} + 1\right)\right)=\rot(2)$ (Equation~\ref{eq:case1c:(m+m')/2+1}), so that, by dividing by two, $\rot\left(-\rho\left(\frac{\frac{m' + m}{2} + 1}{2}\right)\right)=rot(1)$. Since either $m/2$ or $m'/2$ is greater than $a = \frac{\frac{m' + m}{2} + 1}{2}$ we obtain a contradiction: either robot $u_{m - 2a}=\rot(-1)(r_2)$ is a robot located at $\bar{r_2'}$ (which should be empty) or robot $u_{2a}' = \rot(1)(r_1')$ is a robot located at $\bar{r_1}$.

Hence, we know that $m = m'$ and $d_1=d_2$. This implies that $r_1$ and $r_2'$ are symmetric with respect to the line $L_3$, bisector of $L_1$ and $L_2$. Similarly, $s_2(r_1)$ is the symmetric of $s_1(r_2')$ with respect to $L_3$, $s_1(s_2(r_1))$ is the symmetric of $s_2(s_1(r_2'))$, etc...
 The two sets of robots $\{u_i\}_{0\leq i\leq m}$ and $\{u_i'\}_{0\leq i\leq m}$ are symmetric with respect to $L_3$.

If some robots are not in this sets, i.e., $R' = R\setminus \left(\{u_i\}_{0\leq i\leq m} \cup \{u_i'\}_{0\leq i\leq m}\right)\neq \emptyset$, then both $L_1$ and $L_2$ are axes of symmetry of $R'$ (because the configuration is quasi-symmetric for $L_1$ and $L_2$), hence, by Lemma~\ref{lem:2 axis of sym => bisector is an axis}, $L_3$ is also an axis of symmetry of $R'$.
Hence, the whole configuration is node-edge-symmetric with respect to $L_3$, a contradiction.

%
%    CASE 1.d
% 
\textbf{Case $1.d$:} $r_2' = \rot(d_2)(r_2)$ and $r_2$ is the leading robot (see Figure~\ref{fig:case 1.d}).

We first show that $\frac{m'+m}{2}$ is even. 
We have, using the same notation as in case $1.c$, 
\begin{equation}\label{case1.c two equations}
\begin{array}{l}
\rot(d_1-1 + \rho m + d_2-1) = \rot(-\rho) \\
\rot(1-d_1 + \rho m' + 1-d_2) = \rot(-\rho)
\end{array}
\end{equation}
and we obtain by addition
\[
\rot(\rho ((m' + m) + 2)) = \rot(0).
\]
Since $n$ is odd and $m'+m$ is even, we have
\[
\rot\left(\rho \left(\frac{m' + m}{2} + 1\right)\right) = \rot(0)
\]
If we assume for the sake of contradiction that $\frac{m' + m}{2} + 1$ is even, then,
\[
\rot\left(\rho \left(\frac{\frac{m' + m}{2} + 1}{2}\right)\right) = \rot(0)
\]
 and since either $m/2$ or $m'/2$ is greater than $a = \frac{\frac{m' + m}{2} + 1}{2}$, we obtain a contradiction ($u_{2a}$ is collocated with $r_1$ or $u_{2a}'$ is collocated with $r_1'$, and this would create a cycle in the sequence $u_0, u_1,\ldots$, which is not possible).

So we know that $\frac{m' + m}{2}$ is even. We now show that $d_1 = d_2$. We can subtract and divide by 2 the two Equations~(\ref{case1.c two equations}) to obtain
\[
\rot\left(\rho \left(\frac{m' - m}{2}\right)\right) = \rot(d_2+d_1+2)
\]
%\begin{align*}
%\rot\left(\rho \left(m' - m\right)\right) = \rot\left(2d_2+2d_1+4\right)
%\end{align*}
Since $\frac{m' - m}{2}$ is even and $d_2+d_1$ is also even (recall that $d_1$ and $d_2$ are odd), we can again divide by 2:
\begin{align*}
    \rot\left(\rho \left(\frac{m' - m}{4}\right)\right) = \rot\left(\frac{d_2+d_1}{2}+1\right)
\end{align*}
If, by contradiction, $d_1\neq d_2$, then either $d_1$ or $d_2$ is at most $\frac{d_2+d_1}{2}+1$.  This implies that a robot in the sequence $(u_i)_i$ or $(u_i')_i$ is between $r_1$ and $r_1'$ or between $r_2$ and $r_2'$. Assume that $m' > m$ and $d_1\leq \frac{d_2+d_1}{2}+1$ (the other cases are done in a similar way). Then the robot at $u_{2\frac{m'-m}{4}}'$ is between $r_1$ and $r_1'$,  at distance $\frac{d_2+d_1}{2}+1$ from $r_1'$. The only robot that might be in between those robots is a robot that is on the quasi-axis (the target node), so at distance $\frac{d_1 + 1}{2}$ from $r_1'$, but $\frac{d_2+d_1}{2}+1>\frac{d_1 + 1}{2}$ so this is not possible.
%QB: maybe I have to add detail for the above affirmation

So we know $d_1=d_2$ and this implies that $r_1$ and $r_2$ are symmetric with respect to $L_3$, the line bisector of $L_1$ and $L_2$. In turn, $s_2(r_1)$ and $s_1(r_2)$ are also symmetric with respect to $L_3$. $s_1(s_2(r_1))$ and $s_2(s_1(r_2))$ as well and so on. The same thing is true between $r_1'$ and $r_2'$, so $L_3$ is an axis of symmetry of the set of robots $R_{seq} = \{r_1, r_1', u_1, u_1', \ldots, u_m, u_m'\}$.

If some robots are not in this set i.e., $R' = R\setminus R_{seq}\neq \emptyset$, then both $L_1$ and $L_2$ are axes of symmetry of $R'$ (because the configuration is quasi-symmetric for $L_1$ and $L_2$), hence, by Lemma~\ref{lem:2 axis of sym => bisector is an axis}, $L_3$ is also an axis of symmetry of $R'$.

Hence, the whole configuration is node-edge-symmetric with respect to $L_3$.\\
%\textbf{Case $1.b$:} $r_1 = \rot(x+1)(l_1)$ and $r_2 = \rot(-x-1)(l_2)$. The proof is the same as the previous case.
%
%    CASE 2
% 
\textbf{Case 2: $i=2$.}
Then $u_m = r_1'$. Again, since $s_1\circ s_2 = \rot(\rho)$, we have that 
\[
r_1 \apprho{} u_2 \apprho{} u_4 \apprho{} \ldots  \apprho{} u_{m-1} \apprho{} \bar{r_1'}
\]
Where $u_{m-1} \apprho{} \bar{r_1'}$ because $u_m \apps{1} \bar{r_1'}$. So we have $\bar{r_1'} = \rot(\rho m/2)(r_1)$.
Since $r_1$ and $\bar{r_1'}$ are adjacent (more precisely, we assumed w.l.o.g that $\bar{r_1'} = \rot(1)(r_1)$) then $\rho m/2 \equiv 1 \mod n$ and this implies that $\rho$ and $n$ are coprime. Hence, $m$ is the smallest integer such that $\rho m/2 \equiv 1 \mod n$ (otherwise each node contains a robot). 

We can do the same starting from $r_2$, but this time the sequence starts with $s_1$:
\[
r_2 \apps{1} u_1' \apps{2} u_2'\apps{1} \ldots \apps{1} u_k'
\]
where $u_k'=r_2'$, hence we have
\[
r_2 \apprho{-1} u_2' \apprho{-1} u_4' \apprho{-1} \ldots  \apprho{-1} u_{m'-1}'
\apprho{-1} \bar{r_2'}
\]
 Using the same argument as previously, we know that $m' = m$. Thus, $r_2 = \rot(\rho m/2)(\bar{r_2'})$, but since $\rot(\rho m/2)=\rot(1)$, we have $r_2 = \rot(1)(\bar{r_2'})$, which implies that the two quasi-axis have opposite orientations, as depicted in Figure~\ref{fig:case 2 example}.%, as depicted in Figure~\ref{fig:case 2} (and cannot be as depicted in Figure~\ref{fig:case 2'})
 % QB: the following is not true anymore (d_1 != d_2), and $r_1$ and $r_2$ are symmetric with respect to the line $L_3$, bisector of $L_1$ and $L_2$.
 %Similarly, $s_2(r_1)$ is the symmetric of $s_1(r_2)$ with respect to $L_3$, $s_1(s_2(r_1))$ is the symmetric of $s_2(s_1(r_2))$, etc... 
% The two sets of robots $\{u_i\}_{0\leq i\leq m}$ and $\{u_i'\}_{0\leq i\leq m}$ are symmetric with respect to $L_3$.
%If some robots are not in this sets i.e., $R' = R\setminus \left(\{u_i\}_{0\leq i\leq m} \cup \{u_i'\}_{0\leq i\leq m}\right)\neq \emptyset$, then both $L_1$ and $L_2$ are axes of symmetry of $R'$ (because the configuration is quasi-symmetric for $L_1$ and $L_2$), hence, by Lemma~\ref{lem:2 axis of sym => bisector is an axis}, $L_3$ is also an axis of symmetry of $R'$.
%So the whole configuration is symmetric with respect to $L_3$.

So, in this case, it is possible to have two quasi-axis with opposite orientations. We now prove that the configuration cannot have a third quasi-axis. Indeed, if the configuration has a $\{r_3, r_3'\}$-quasi-axis $L_3$, then it has an orientation opposite to either $\{r_1,r_1'\}$ or $\{r_2,r_2'\}$. Assume w.l.o.g that the orientation of $\{r_3, r_3'\}$ is opposite with $\{r_1, r_1'\}$. This means that we can do the same reasoning between these two quasi-axis as between $\{r_1,r_1'\}$ and $\{r_2,r_2'\}$, using the $s_3$ symmetry of axis $L_3$ instead of the $s_2$ symmetry. In this reasoning, we have two sequences of robots:
\begin{align*}
r_1 = v_0 \apps{3} v_1 \apps{1} v_2\apps{3} \ldots \apps{i} v_{m_{1,3}}\\
r_1' = v'_0 \apps{3} v'_1 \apps{1} v'_2\apps{3} \ldots \apps{i} v'_{m'_{1,3}}
\end{align*}
However, only {\bf Case 2} is possible (because we saw that if {\bf Case 1} applies, the two quasi-axis have the same orientation). Hence, by denoting $\rho_{1,3} = s_1\circ s_3$, we have
\begin{align*}
    \bar{r_1'} &= \rot(\rho_{1,3} m_{1,3}/2)(r_1)\\
    r_3 &= \rot(\rho_{1,3} m_{1,3}/2)(\bar{r_3'})
\end{align*}
and $m_{1,3}$ is the smallest integer such that $\rho_{1,3} m_{1,3}/2 \equiv 1 \mod n$. One can see that at least one robot is either $(i)$ in the sequence $(u_i)_i\cup(u'_i)_i$ but not in $(v_i)_i\cup(v'_i)_i$ or the contrary $(ii)$ in the sequence $(v_i)_i\cup(v'_i)_i$ but not in $(u_i)_i\cup(u'_i)_i$. Indeed, otherwise $m_{1,3} = m$, which implies $\rho = \rho_{1,3}$, which implies $L_2 = L_3$, a contradiction.
Assume case $(i)$ without loss of generality. So a robot $r$ is not $(v_i)_i\cup(v'_i)_i$. When we apply sequentially $\rho_{1,3}$, we never obtain a robot in $(v_i)_i\cup(v'_i)_i$, so the set of robot $(\rho_{1,3}^i(r))_i$ is finite and periodic. However, since $\rho_{1,3}$ and $n$ are coprime, this is not possible. So, there is at most two quasi-axis of symmetry with opposite orientation. \qed

\begin{figure}
    \centering
\begin{tikzpicture}
\tikzstyle{robot}=[circle,fill=black!100,inner sep=0.05cm]
  
\def\n{25}
\def\radius{2.4cm}
\confRing

\confAxis[dashed]{0}{-4}{20}
\confAxis[dashed]{6}{-4}{20}

\confRobot[]{1}
\confRobot[red]{-2}
\confRobot[]{4}
\confRobot[]{-4}

\confRobot[red]{9}
\confRobot[]{-9}

\confRobot[]{11}
\confRobot[]{-11}

\end{tikzpicture} 
    \caption{Case 2: a configuration with two quasi-axis with opposite orientations. The crashed robot is one of the red robots.}
    \label{fig:case 2 example}
\end{figure}
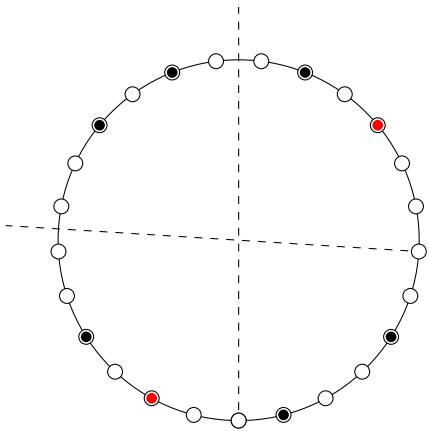

\end{proof}

%%%%%%%%%%%%%%%%%%%%%%%%%%%%%%%%%%%%%%%%%%%%%%%%%%%
%
%      NE(k) -> L3
%
%%%%%%%%%%%%%%%%%%%%%%%%%%%%%%%%%%%%%%%%%%%%%%%%%%%
\subsection{From $\NE(k)$, $k>4$, to $\LThree$}\label{sec:NE->L3}

In this section, first we show in Lemma~\ref{lem:NE(k) -> NE(k'), without crash}, when no moving robot is crashed, the configuration remains node-edge-symmetric. In Lemma~\ref{lem:NE(k) -> L3, without crash}, we prove in this case that eventually the number of occupied nodes decreases. Then, we show that, when a crash occurs among the secondary robots, we recover in one round the same configuration that we would have reached if no crash occurred (Lemma~\ref{lem:NE(k) -> NE(k'), secondary crash}). Finally, we show that, when a crash occurs among the main robots, we eventually reach a configuration where the main robots are correct (Lemma~\ref{lem:NE(k) -> NE(k'), main crash}).

\begin{lemma}
%\begin{theoremEnd}[disc]{lemma}
\label{lem:NE(k) -> NE(k'), without crash}
Let $C \in \NE(k) \setminus P$, with $k>4$ robots, and $C\oneRound C'$. If all moving robots reach their destination (no robot ordered to move is crashed), then $C'\in \NE(k') \setminus P$, with $k'= k$, $k' = k-1$, or $k'=k-2$. 
Also, if $k' \leq 4$, then $C\in \LFive$ and $C'\in \LThree$.
%\SK{Is it clear that, if $k>4\land k'\leq4$, then $C\in \LFive$? QB: Should be better now}
%\end{theoremEnd}
\end{lemma}
\begin{proof}
%\begin{proofEnd}
Clearly, the obtained configuration $C'$ remains node-edge-symmetric, because the main (symmetric) robots are ordered to move towards the same target node located on the axis of symmetry, and the secondary (symmetric) robots might be ordered to move towards (or away) the target node.
In the following, we show that 
$C'$ is not periodic.

If only the main robots are ordered to move in $C$, then, from the condition of Line~\ref{algo:ring:line check not periodic}, we know that $C'$ is not periodic. If they reach the target node, then $k'=k-1$ or $k'=k-2$ (depending whether the target node is occupied or not in $C$), otherwise $k'=k$.
In the case $C\in \LFive$ (\LFiveFigure{}), only the main robots are ordered to move, and we obtain $C'\in \LThree$ (\LThreeFigure{}). If $k' \leq 4$, then it means the main robots have reached $v_{target}$ (it is the only node where a multiplicity can be created). If $v_{target}$ is empty in $C$, then $k'=k-1$, but $k$ is even ($C$ is node-edge-symmetric with no robot on the axis) and so $k\geq 6$, which implies $k' > 4$. So if $k'\leq 4$, it means $v_{target}$ is occupied in $C$, $k = 5$ and $k' = k -2 = 3$. When $k=5$, the main robots are ordered to move towards the target only when the configuration is $\LFive$.

If only the secondary robots move in $C$, from the condition of Line~\ref{algo:45}, it means that either there are 4 occupied nodes (so $C'$ is not periodic because $n$ is odd) or there are 5 occupied nodes but 3 of them are adjacent, so $C'$ is not periodic, and $k'=k$. %\QB{This case should actually be treated later because here we assume k > 5 now}.

Otherwise, the main and the secondary robots move in $C$. In this case, assume by contradiction that $C'$ is periodic. Since the secondary robots are ordered to move in $C$, then we know that $C'$ is periodic if they stay idle, so $C'$ is periodic and the configuration $C''$ obtained from $C'$ by moving the two secondary robots to their original position is also periodic. %\SK{I'm sorry but I cannot understand this part. The secondary robots stay idle, so why they move to original node?} \QB{We don't actually move the robots, the sentence is more like "if we move the robots, the configuration remains periodic}
Since $C'$ is periodic with period at least $3$ (because $n$ is odd), then there are at least $6$ occupied positions. By moving 2 robots to obtain $C''$, $4$ occupied positions remain unchanged. Since $4\geq 6/2$, by Lemma~\ref{lem:periodic configurations are determined by k/3 consecutive positions}, we have the two configurations, $C'$ and $C''$, are equal, which is a contradiction, so $C'$ is not periodic.\qed
%\SK{One of the "two configurations" may be $C'$, but what is the other?. } \QB{It should be better now}
% a configuration that is periodic with period $p>1$ remains $m$-periodic if robots in a multiple of $m$ occupied location are moved, where $m$ is a divisor of $p$. Here $m\geq3$ because $n$ is odd.
%\end{proofEnd}
\end{proof}

\begin{theorem}[disc]{lemma}\label{lem:NE(k) -> L3, without crash}
Let $C \in\NE(k) \setminus P$ with $k>4$ robots. %\QB{I changed a little bit the statement}
%When executing Algorithm~\ref{algo:ring} from $C \in NE(k) \setminus P$ with $k>4$, 
If no moving robot crashes, then %after a finite number of rounds, we obtain configuration 
there exists $C'\in \NE(k') \setminus P$, $C\overset{*}{\oneRound} C'$, with $k' = k-1$ or $k'=k-2$.  Also, if $k' \leq 4$, then $C\in \LFive$ and $C'\in \LThree$.
\end{theorem}
\begin{proof}
%Since $k>4$, only the main robots are ordered to move, so each moving robot reaches its destination. %\SK{Cannot secondary robots move? QB: yes indeed}
After each round, by Lemma~\ref{lem:NE(k) -> NE(k'), without crash}, %lem:initial config remains initial without crash}, 
the configuration is in $NE(k') \setminus P$ (with $k'=k$, $k' = k-1$ or $k'=k-2$). The main robots do not move towards an occupied node, except when it is $v_{target}$. If they reach $v_{target}$, either $k' = k-1$ (if $v_{target}$ is empty in $C$), or $k' = k-2$ (otherwise). If they are not adjacent to $v_{target}$, they remain the main robots and are one hop closer to $v_{target}$, so they eventually reach $v_{target}$ after a finite number of rounds.\qed
\end{proof}

%\begin{theoremEnd}[disc]
\begin{lemma}\label{lem:NE(k) -> NE(k'), secondary crash}
Let $C \in \NE(k) \setminus P$ with $k>4$ robots, and $C\oneRound C'$. If a secondary robot is crashed in $C$, then $C'\oneRound C''$ such that $C\oneRoundNoCrash C''$.
%\end{theoremEnd}
\end{lemma}
\begin{proof}
%\begin{proofEnd}
When secondary robots are supposed to move and a secondary robot crashes in $C$, then the obtained configuration $C'$ is quasi-periodic, and has a single quasi-axis of symmetry (Lemma~\ref{lem:even quasi-axis is unique}). So it cannot be node-edge-symmetric (otherwise it would have another quasi-axis). So, after executing the algorithm for one more round, all the robots are ordered to move in the opposite orientation with respect to the quasi-axis. All the robots move except the crashed robot by applying \MoveSame, so the obtained configuration $C''$ is similar to the configuration obtained by moving the crashed robot towards the target node, i.e., $C''$ is the configuration obtained from $C$ if no robot is crashed.\qed
\end{proof}
%\end{proofEnd}

In the next Lemma, we show that, if a main robot is crashed, three cases can occur. Each case implies that the algorithm makes progress. The last two cases are easy to consider: $(b)$ when the obtained configuration remains node-edge symmetric but the new main robots are correct, and $(c)$ when it is not node-edge-symmetric, so it is quasi-node-edge symmetric and the robots detect that a crash occurred and recover (after one more round the configuration is the same as is no crash occurred). In the first case $(a)$ however, we can only say that every time the crash of a main robot creates a new node-edge symmetric configuration where the crashed robot remains one of the main robots, the distance between two given robots decreases. This implies that this case can only occur a finite number of times. 
When the crashed robot is a main robot, then there are only two possible choice for the other main robot. If $k$ is even, the other main robot is one of the two closest robots in both directions. If $k$ is odd, the other main robot is one of the two second closest robots in both directions. In this Lemma, we denote by $M(C)$ the set of the two possible main robots defined above. 
%\SK{Can we define $M(C)$ in the lemma statement more clearly? QB: we could but I am afraid it would be too long, as you want.}

%\begin{theoremEnd}[disc]
\begin{lemma}\label{lem:NE(k) -> NE(k'), main crash}
Let $C \in \NE(k) \setminus P$ with $k>4$ robots, and $C\oneRound C'$. If a main robot is crashed in $C$, then either $(a)$ $C'\in \NE(k)\setminus \Periodic$ and the distance between the two robots in $M(C)$ decreases, $(b)$  $C'\in \NE(k)\setminus \Periodic$ and the main robots are correct, or $(c)$ $C'\oneRound C''$ such that either $C\oneRoundNoCrash C''$ or the main robots are correct in $C''$.
%\end{theoremEnd}
\end{lemma}
%\begin{proofEnd}
\begin{proof}
As in Lemma~\ref{lem:NE(k) -> NE(k'), secondary crash}, if the configuration $C'$ is not node-edge symmetric, then either all the quasi-axes have the same orientation and move \MoveSame is executed, or there are two quasi-axes with opposite orientation and move \MoveOpposite is executed. If move \MoveSame is executed, the obtained configuration $C''$ is similar to the configuration obtained from $C$ if no robot is crashed. If move \MoveOpposite is executed, the obtained configuration $C''$  has a new axis of symmetry where the main robots are correct. So, the case $(c)$ holds. %Also, in the former case, the distance between the robots in $M(C)$ decreases (since one of them moves towards the crashed robot).

If $C'$ remains node-edge symmetric, then either the new main robots are correct (case $(b)$), or the correct main robot is one hop closer to the crashed robot (case $(a)$), hence to the other robot in $M(C)$ as well, so the distance between the two robots in $M(C)$ has decreased.\qed
\end{proof}

From Lemmas~\ref{lem:NE(k) -> NE(k'), without crash}--\ref{lem:NE(k) -> NE(k'), main crash}, we obtain the following results.
\begin{theorem}[disc]{lemma}\label{lem:NE(k) -> L3}
When executing Algorithm~\ref{algo:ring} from $C \in \NE(k) \setminus P$ with $k>4$ then, after a finite number of rounds, we obtain configuration $C'\in \NE(k') \setminus P$, $k' = k-1$ or $k'=k-2$ even if there is a crashed robot. Also, if $k' \leq 4$, then $C\in \LFive$ and $C'\in \LThree$.
\end{theorem}

\begin{proof}
    If the main robots are correct in $C$, then Lemma~\ref{lem:NE(k) -> L3, without crash} implies that we eventually reach $C'\in \NE(k') \setminus P$, $k' = k-1$ or $k'=k-2$. If one of the main robots is crashed in $C$, then either the sequence of configurations alternates between node-edge symmetric and quasi-node-edge symmetric configuration, but in this case, by Lemmas~\ref{lem:NE(k) -> NE(k'), secondary crash}--\ref{lem:NE(k) -> NE(k'), main crash}, we know that the sub-sequence of node-edge symmetric configuration is the same as if both main robots were correct (and so we reach the same configuration as before), or every time two node-edge symmetric configurations occur consecutively in the execution, the sum of the distances between the crashed robot and the two closest robots decreases, so this can only occur a finite number of times.\qed
\end{proof}

%%%%%%%%%%%%%%%%%%%%%%%%%%%%%%%%%%%%%%%%%%%%%%%%%%%
%
%      NE(4) -> L3
%
%%%%%%%%%%%%%%%%%%%%%%%%%%%%%%%%%%%%%%%%%%%%%%%%%%%
\subsection{From $\NE(4)$ to $\LThree$}\label{sec:NE(4)->L3}

Starting from a configuration in $\NE(4)$, we obtain a similar results: we prove that if no moving robot crashes, we reach configuration $\LThree$ (Lemmas~\ref{lem:NE(4) -> NE(4) main at distance 2, without crash},~\ref{lem:NE(4) -> L4, secondary move and are correct} and \ref{lem:L4 -> L3}), and then we deal with cases where a crash occurs (Lemmas~\ref{lem:NE(4) -> L4, with main crash}--\ref{lem:NE(4) secondary adjacent -> NE(4) non adjacent, with crash} and \ref{lem:L4 -> L3}). However, here, we have one more case as the secondary robots might move not only to avoid a periodic configuration, but also to ensure that the only configuration reached in $\NE(3)$ is $\LThree$. The case $n=5, 7$ can be done by simple case analysis, so assume in this subsection that $n > 7$.

\begin{theorem}[disc]{lemma}
\label{lem:NE(4) -> NE(4) main at distance 2, without crash}
When executing Algorithm~\ref{algo:ring} from $C \in \NE(4) \setminus P$ where the main robots are at distance more than 2, if the main robots do not crash, then, after a finite number of rounds, we obtain configuration $C'\in \NE(4) \setminus P$, where the main robots are at distance 2.
\end{theorem}
\begin{proof}
    The proof is similar as for Lemma~\ref{lem:NE(k) -> NE(k'), without crash} %lem:NE(k) -> N(k'), k'<k, without crash} 
    (no moving robot is crashed because, with 4 occupied nodes, no periodic configuration can be reached so only the main robots are ordered to move), but we stop before the main robots reach the target node, hence remain at distance $2$ from one another.\qed
\end{proof}

\begin{theorem}[disc]{lemma}\label{lem:NE(4) -> L4, secondary move and are correct}
When executing Algorithm~\ref{algo:ring} from $C \in \NE(4) \setminus P$, if the secondary robots are ordered to move and do not crash, then, after a finite number of rounds, we obtain configuration $C'\in \LFour$.
\end{theorem}
\begin{proof}
    Since the configuration cannot be periodic with only 4 occupied nodes ($n$ is odd), if the secondary robots are ordered to move, it means that the two main robots are at distance $2$ from each other. Hence, the secondary robots move towards the target node until we reach $\LFour$ (\LFourFigure{}).\qed
\end{proof}

%\begin{theoremEnd}[disc]
\begin{lemma}
\label{lem:NE(4) -> L4, with main crash}
When executing Algorithm~\ref{algo:ring} from $C \in \NE(4) \setminus P$ where the main robots are at distance more than 2, if a main robot crashes, then, after a finite number of rounds, we obtain either a configuration in $\LFour$ or a configuration in $\NE(4) \setminus P$ where the axis of symmetry changed from $C$.
\end{lemma}
%\end{theoremEnd}
%\begin{proofEnd}
\begin{proof}
By the definition of $C$, the main robots are ordered to move.
When a main robot crashes, either the obtained configuration $C'$ is node-edge-symmetric or not. 
If $C'$ is not node-edge symmetric, so is quasi-symmetric, %with only one quasi-axis of symmetry because $k=4$ and $n$ is odd, let $C'$ be the configuration.
%Then, in $C'$, the leading robot is the other main robot which is not crashed.
and by applying $\MoveSame$ or $\MoveOpposite$, 
%Then, all the robots move in orientation opposite to the quasi-axis by the execution of Line~\ref{alg:all}, and 
the obtained configuration $C''$ is in $\NE(4)\setminus P$ where the distance between main robots decreased by 2. 
In the former case, if $C'$ is node-edge-symmetric, then either the new pair of main robots is correct (not crashed) and the Lemma is proved, or the new pair still contain the crashed robot. From there, we apply the same reasoning and either we reach $\LFour$ (\LFourFigure{}) or again the pair of main robots changes. If this occurs the crashed robot cannot be in the pair of main robots again, so the Lemma is proved. \qed %(\QB{another way to prove it, maybe simpler, is to do it like in the NE(k>4) case, ie, every time the pair of main robots changes, the sum of the distance from the crashed robot with the two closest robots decreases}).
%\end{proofEnd}
\end{proof}

%\begin{theoremEnd}[disc]
\begin{lemma}\label{lem:NE(4) -> NE(4), with secondary crash}
When executing Algorithm~\ref{algo:ring} from $C \in \NE(4) \setminus P$ where the secondary robots are ordered to move, if the secondary robots are not adjacent, let $C'$ be a configuration such that $C\oneRoundSecCrash C'$. Then, $C'\oneRound C''$ such that $C\oneRoundNoCrash C''$. % or the main robots are correct in $C''$.
%\end{theoremEnd}
\end{lemma}
\begin{proof}
%\begin{proofEnd}
Because the secondary robots are ordered to move in $C$, $k=4$ and $n$ is odd, $C\not\in \LFour$ and the main robots are adjacent to $v_{target}$.
Then the secondary robots are not adjacent to the main robots, and move towards $v_{target}$.
Thus, in $C''$, the distance between secondary robots increased by 2 and between a secondary robot and a main robot decreased by 1 from $C$.

When a secondary robot crashes, $C'$ is not node-edge symmetric (because the distance between the secondary robots is strictly greater than 2 as they were not adjacent in $C$) and is quasi-node-edge-symmetric 
%\SK{"not node-edge symmetric" and "quasi-symmetric" are not same? QB: a  conf can be both, but here it is only quasi-node-edge symmetric}% with only one quasi-axis of symmetry.
%In $C'$, the gap distance associated with the quasi-axis is even, and the reading robot is the crashed robot.
%Then, all the robots move in orientation opposite to the quasi-axis by the execution of Line~\ref{alg:all},
and, after executing $\MoveSame$ or $\MoveOpposite$, the configuration is in $\NE(4)\setminus P$ where the distance between secondary robots increased by 2 and the distance between a secondary robot and a main robot decreased by 1 from $C$, that is, it is $C''$.

Thus, the Lemma holds.\qed
%\end{proofEnd}
\end{proof}

%\begin{theoremEnd}[disc]
\begin{lemma}\label{lem:NE(4) secondary adjacent -> NE(4) non adjacent, with crash}
When executing Algorithm~\ref{algo:ring} from $C \in \NE(4) \setminus P$ where the secondary robots are ordered to move to obtain $C'$, if the secondary robots are adjacent and one of them crashes, then, $C'$ is in $\NE(4)\setminus \Periodic$ where the main robots are correct and the secondary robots are not adjacent.
\end{lemma}
%\end{theoremEnd}
%\begin{proofEnd}
\begin{proof}
If the secondary robots are ordered to move, this means that the main robots are adjacent to the target node. If the secondary robots are adjacent and one of them is crashed, then, after one round we obtain node-edge symmetric configuration $C'$ where the two previous secondary robots are symmetric to the two previous main robots. In $C'$, the new two secondary robots consist in a previous main robot and a previous secondary robot. Hence they are adjacent if the number of nodes is exactly 7. If there are 7 (or even 5) nodes, the secondary robots are not ordered to move as no periodic configuration can be created (condition at Line~\ref{algo:45} does not apply).\qed
%\end{proofEnd}
\end{proof}

\begin{comment}
\begin{tikzpicture}
\tikzstyle{robot}=[circle,fill=black!100,inner sep=0.05cm]
  
\def\n{7}
\def\radius{1cm}

\confRing

\confAxis{0}{0}{20}

\confRobot[red]{-1}
\confRobot[red]{1}

\confRobot[blue]{-3}
\confRobot[blue]{3}

\end{tikzpicture}
\end{comment}

\begin{theorem}[disc]{lemma}\label{lem:L4 -> L3}
When executing Algorithm~\ref{algo:ring} from $C \in \LFour$, we reach a configuration in $\LThree$ in at most two rounds.
\end{theorem}
\begin{proof}
%\SK{Please check this.}
A configuration $C$ in $\LFour$ has four occupied positions, there are two blocks of size two separated by one empty node (\LFourFigure{}). 
Then, because it is node-edge symmetric configuration, the empty node is on the axis, i.e., $v_{target}$.
Now, Algorithm~\ref{algo:ring} commands, in such a situation, the two main robots to move to $v_{target}$ in $C$. Let $C''$ be the obtained configuration from $C$ after one round. 
\begin{itemize}
\item None of the two main robots crashes. In $C''$, two main robots of $C$ are on $v_{target}$. It is $\LThree$ (\LThreeFigure{}).
\item One of the main robots crashes. As the other one moves to $v_{target}$, $C''$ is $\LFour'$ (\LFourPrimeFigure{}), there is a block of size three, followed by an empty node $u$ and a single occupied node. Let $u_0$ be the other neighboring node of a extremal node of the block than $u$.
Then, by applying \MoveSame, the non-crashed robots move toward $u_0$, thus, the obtained configuration becomes $\LThree$ such that the crashed robot is on a new target node.  
\end{itemize}
So, starting from a configuration in $\LFour$, a configuration in $\LThree$ is reached in at most two rounds.\qed
\end{proof}

%%%%%%%%%%%%%%%%%%%%%%%%%%%%%%%%%%%%%%%%%%%%%%%%%%%
%
%      LE(3) -> gathered
%
%%%%%%%%%%%%%%%%%%%%%%%%%%%%%%%%%%%%%%%%%%%%%%%%%
\subsection{From $\LThree$ to the gathered configuration}\label{sec:L3->Gathered}

\begin{theorem}[disc]{lemma}\label{lem:L2}
When executing Algorithm~\ref{algo:ring} starting from a configuration in $\LTwo$ and one occupied location only hosts crashed robots, a gathered configuration is reached in one round.
\end{theorem}

\begin{proof}
A configuration in $\LTwo$ (\LTwoFigure{}) has two adjacent occupied locations, and one of them only hosts crashed robots. Now, Algorithm~\ref{algo:ring} commands that in such a situation, non-crashed robots move to the adjacent location. So, in one round, all non-crashed robots joint the location hosting crashed robots (that obviously didn't move). As a result, a single location is now occupied, and Algorithm~\ref{algo:ring} does not command any move in this situation. So, the configuration is gathered. \qed
\end{proof}

\begin{theorem}[disc]{lemma}\label{lem:gather}
When executing Algorithm~\ref{algo:ring} starting from a configuration in $\LThree$, a gathered configuration is reached in 4 rounds.
\end{theorem}

\begin{proof}
A configuration in $\LThree$ has three occupied positions, the two extremal are at distance $4$ from one another, and the third location is in the middle (\LThreeFigure{}). As a result, in Algorithm~\ref{algo:ring}, the target $v_{target}$ is the middle location (that lies on a axis of symmetry), and the main robots are the ones located on the extremal locations.
We consider three cases:
\begin{enumerate}
\item There is no crashed robot. In that case, the extremal robots move toward the central occupied location, and the distance between the extremal occupied locations becomes two (we obtain configuration $\LThree''$ \LThreeFigurePS{}). After one more round, all robots gather at the central occupied location.
\item At least one robot crashes on the central occupied location. The extremal robots all move inwards toward the central occupied location, and after two rounds, a gathered configuration is reached.
\item At least one robot crashes on one (say the left one) extremal occupied location. In that case, %we may reach different situations:
%\begin{itemize}
%\item If all robots on the left extremal location are crashed, 
we obtain after one round a configuration such that there is one left extremal location with only crashed robots, one empty node, and two occupied nodes (configuration $\LThree'$ \LThreeFigureP{}). After one round, Algorithm~\ref{algo:ring} commands all non-crashed robots to go left by the condition of Line~\ref{alg:all} (except if $n=5$ but in this case the crashed robot is on the target node). As a result, we now have a node-edge symmetric configuration $\LThree''$ (\LThreeFigurePS{}).
After one round, all robots on the right extremal move to the middle node. We thus reach a configuration in $\LTwo$ (\LTwoFigure{}) where the occupied node on the left only hosts crashed robots, from which gathering is achieve in one round by Lemma~\ref{lem:L2}. Overall, after four rounds, a gathered configuration is attained.
%\item If a subset of robots on the left extremal location are crashed, we obtain after one round a configuration such that there is one left extremal configuration with only crashed robots, and three locations occupied by non-crashed robots. In that case, the target node become opposite to the occupied locations, and the algorithms commands the extremal robots to move away. Since the left location only has crashed robots, only the robots on the extremal right location move away. So, after one round, we reach a configuration such that, the left extremal location only contains crashed robots, then two occupied locations, one free node, and one occupied location. This configuration has two quasi axes with opposite orientation, so the robots on the left extremal location and the robots on its neighboring location are requested to move apart. The crashed robots do not move, but the neighboring robots join the next occupied location. So, after one round, we reach a configuration that is in $\LThree$, but with all robots on the left extremal position crashed. From the previous case, with five more rounds, a gathered configuration is achieved. In total, after $8$ rounds, a gathered configuration is reached.
%\end{itemize}
\end{enumerate}
So, starting from a configuration in $\LThree$, a gathered configuration is reached within $4$ rounds.\qed
\end{proof}

From the previous Lemmas in Subsect.~\ref{sec:NE->L3}--\ref{sec:L3->Gathered}, we obtain the following results.
\begin{theorem}[disc]{theorem}
Starting from a configuration $C\in\NE(k)\setminus\Periodic$ ($k>3$), Algorithm~\ref{algo:ring} solves the SUIG problem on ring-shaped networks without multiplicity detection in FSYNC.
\end{theorem}
\begin{proof}
When executing Algorithm~\ref{algo:ring} from a configuration $C\in\NE(k)\setminus\Periodic$ with $k>4$, by Lemma~\ref{lem:NE(k) -> L3}, after a finite number of rounds, we obtain a configuration in $\LThree$ or $\NE(4)$ even if there is a crashed robot.
If it is in $\NE(4)$, by Lemmas~\ref{lem:NE(4) -> NE(4) main at distance 2, without crash}--\ref{lem:NE(4) secondary adjacent -> NE(4) non adjacent, with crash}, we obtain configuration in $\LFour$.
By Lemma~\ref{lem:L4 -> L3}, the configuration in $\LFour$ becomes a configuration in $\LThree$.
After that, from a configuration in $\LThree$, by Lemma~\ref{lem:gather}, a gathering is achieved.\qed
\end{proof}
%\textcolor{red}{QB: TODO add figures to the previous proof.}

\section{Concluding Remarks}
\label{sec:conclusion}

We further characterized the solvability of the stand-up indulgent rendezvous and gathering problems on ring-shaped networks by Look-Compute-Move oblivious robots. A number of open questions are raised by our work:
\begin{itemize}
\item Is it possible to extend our possibility result for FSYNC SUIR in node-node symmetric configurations to general SUIG?
\item Is it possible to solve FSYNC SUIG starting from a non-periodic, non-symmetric configuration?  
\item Aside from line-shaped networks (already studied by Bramas et al.~\cite{techreport_line}), is the problem solvable in other topologies?
\item Can additional capabilities help the robots solve the problem in the cases we identified as impossible? 
\end{itemize}
\newpage
\bibliography{ref}

\begin{thebibliography}{10}

\bibitem{ND2006}
Noa Agmon and David Peleg.
\newblock Fault-tolerant gathering algorithms for autonomous mobile robots.
\newblock {\em SIAM Journal on Computing}, 36(1):56--82, 2006.

\bibitem{TPRLSX19}
Thibaut Balabonski, Pierre Courtieu, Robin Pelle, Lionel Rieg, S{\'{e}}bastien
  Tixeuil, and Xavier Urbain.
\newblock Continuous vs. discrete asynchronous moves: {A} certified approach
  for mobile robots.
\newblock In {\em Networked Systems - 7th International Conference, {NETYS}
  2019, Marrakech, Morocco, June 19-21, 2019, Revised Selected Papers}, number
  11704 in LNCS, pages 93--109. Springer, 2019.

\bibitem{ZSS2013}
Zohir Bouzid, Shantanu Das, and S\'ebastien Tixeuil.
\newblock Gathering of mobile robots tolerating multiple crash faults.
\newblock In {\em IEEE 33rd International Conference on Distributed Computing
  Systems (ICDCS)}, pages 337--346, 2013.

\bibitem{techreport_line}
Quentin Bramas, Sayaka Kamei, Anissa Lamani, and S\'ebastien Tixeuil.
\newblock Stand-up indulgent gathering on lines.
\newblock In {\em Stabilization, Safety, and Security of Distributed Systems.
  SSS 2023}, Lecture Notes in Computer Science, 2023.

\bibitem{QAS2020}
Quentin Bramas, Anissa Lamani, and S\'ebastien Tixeuil.
\newblock Stand up indulgent rendezvous.
\newblock In {\em Stabilization, Safety, and Security of Distributed Systems.
  SSS 2020}, 2020.

\bibitem{QAS2021}
Quentin Bramas, Anissa Lamani, and S\'ebastien Tixeuil.
\newblock Stand {U}p {I}ndulgent {G}athering.
\newblock In {\em Algorithms for Sensor Systems. ALGOSENSORS 2021}, number
  12961 in LNCS, 2021.

\bibitem{BLT23j}
Quentin Bramas, Anissa Lamani, and S{\'{e}}bastien Tixeuil.
\newblock Stand up indulgent gathering.
\newblock {\em Theor. Comput. Sci.}, 939:63--77, 2023.
\newblock \href {https://doi.org/10.1016/j.tcs.2022.10.015}
  {\path{doi:10.1016/j.tcs.2022.10.015}}.

\bibitem{QS2015}
Quentin Bramas and S\'ebastien Tixeuil.
\newblock Wait-free gathering without chirality.
\newblock In {\em 22nd Structural Information and Communication Complexity
  (SIROCCO)}, number 9439 in LNCS, pages 313--327, 2015.

\bibitem{XMP2020}
Xavier D\'efago, Maria Potop-Butucaru, and Philippe Raipin-Parv\'edy.
\newblock Self-stabilizing gathering of mobile robots under crash or byzantine
  faults.
\newblock {\em Distributed Computing}, 33:393--421, 2020.

\bibitem{PGN2019}
Paola Flocchini, Giuseppe Prencipe, and Nicola Santoro, editors.
\newblock {\em {D}istributed {C}omputing by {M}obile {E}ntities, {C}urrent
  {R}esearchin {M}oving and {C}omputing}.
\newblock Number 11340 in LNCS. Springer, 2019.

\bibitem{REA2008}
Ralf Klasing, Euripides Markou, and Andrzej Pelc.
\newblock Gathering asynchronous oblivious mobile robots in a ring.
\newblock {\em Theoretical Computer Science}, 390(1):27--39, 2008.

\bibitem{SY1999}
Ichiro Suzuki and Masafumi Yamashita.
\newblock Distributed anonymous mobile robots: Formation of geometric patterns.
\newblock {\em SIAM Journal on Computing}, 28(4):1347--1363, 1999.

\end{thebibliography}

%\newpage
%\appendix

%\section{Omitted Proofs}

%\printProofs

\end{document}